\documentclass{article}
\usepackage[utf8]{inputenc}
\usepackage{amsmath}
\usepackage{amsthm}
\usepackage{amsfonts}
\usepackage{mathtools}
\usepackage{amsmath,amssymb}
\usepackage{graphicx}
\usepackage[vcentermath]{youngtab}
\usepackage{mathtools}
\usepackage{latexsym,amssymb,amsmath,amsfonts, amsthm, xcolor}
\usepackage{graphicx}
\usepackage{tikz}
\usetikzlibrary{matrix,arrows}

\newtheorem{theorem}{Theorem}
\newtheorem{lemma}[theorem]{Lemma}

\newtheorem{proposition}[theorem]{Proposition}
\newtheorem{example}{Example}
\newtheorem{remark}[theorem]{Remark}

\title{Subspace Controllability and Clebsch-Gordan Decomposition of Symmetric Quantum Networks}

\author{D. D'Alessandro\thanks{Department of Mathematics, Iowa State University, Ames, IA 50011, daless@iastate.edu}}

\begin{document}

\maketitle

\begin{abstract}
We describe  a framework for the controllability analysis of networks of $n$  quantum systems of an arbitrary dimension $d$, {\it qudits}, with  dynamics determined by Hamiltonians that are invariant under the permutation group $S_n$. Because of the symmetry, the underlying Hilbert space,  ${\cal H}=(\mathbb{C}^d)^{\otimes n}$,  splits into invariant subspaces for the Lie algebra of $S_n$-invariant elements in $u(d^n)$, denoted here by $u^{S_n}(d^n)$. The dynamical Lie algebra ${\cal L}$, which determines the controllability properties of the system,  is a Lie subalgebra of such a Lie algebra $u^{S_n}(d^n)$. If ${\cal L}$ acts as $su\left( \dim(V) \right)$ on each of the invariant subspaces $V$, the system is called {\it subspace controllable}.  Our approach is based on recognizing that such   a splitting of the Hilbert space ${\cal H}$ coincides with the {\it Clebsch-Gordan} splitting of $(\mathbb{C}^d)^{\otimes n}$ into {\it irreducible representations} of $su(d)$. In this view,  $u^{S_n}(d^n)$, is the direct sum of certain $su(n_j)$ for some $n_j$'s we shall specify, and its {\it center}  which is the Abelian (Lie) algebra generated by the {\it Casimir operators}. Generalizing the situation in \cite{confraJMP} and \cite{Marco},   we consider  dynamics with arbitrary local simultaneous control on the qudits 
and a symmetric two body interaction. Most of the results presented are for general $n$ and $d$ but we recast the results of \cite{confraJMP}, \cite{Marco} in this new general framework and provide a complete treatment and proof of subspace controllability for the new case of $n=3$, $d=3$, that is,  {\it three qutrits}. Our results are motivated by recent great interest in symmetric quantum states and systems (see, e.g.,  \cite{PIPPO1}, \cite{Harrow}, \cite{Migdal}) both at the theoretical and experimental level and by recent proposals in Geometric Quantum Machine Learning \cite{Marco1} \cite{Marco2} to exploit symmetries in the data and  in quantum circuits to improve the performance of  learning protocols.  
\end{abstract}

\section{Introduction}\label{Intro}

It  has been known since the beginning of the (modern) theory of quantum control (see, e.g., \cite{Tarn},  \cite{Jurde}, \cite{Murti})
that the controllability of closed, finite dimensional, quantum systems can  be assessed by examining the Lie algebra generated by the available Hamiltonians. Furthermore, it  follows from classical analysis results that controllability is {\it generic}, that is, almost any pair of Hamiltonians generate the whole Lie algebra $su(N)$ which implies that any special unitary evolution can be performed  by the system at hand (see, e.g., \cite{mikobook} for a proof).  In spite of the generic nature   of controllability, uncontrollable systems are of great interest both in theory and in applications. Uncontrollable quantum systems occur most often as a consequence of the presence of a {\it symmetry group}, that is, a  group which commutes with all the available Hamiltonians for the system.  Systems which admits symmetries may naturally arise as physical systems (e.g., a system of undistinguishable bosons) or may be artificially built  as for example in {\it Geometric Quantum Machine Learning} protocols (see, e.g., \cite{Marco2}) where one would like to take advantage of the symmetry of (quantum) data sets to overcome some of the limitations of quantum machine learning \cite{Marco1}. For these and other reasons, interest in symmetric quantum states and systems has been very large  in the last few years and the subject  of the work of many groups (see, e.g., \cite{Aloy}, \cite{Harrow}, \cite{Larocca}-\cite{Marco2}, \cite{Skolik}).

If a quantum system of dimension $N$ admits a finite group of symmetries $G$,  the largest possible dynamical Lie algebra is $u^{G}(N)$, defined as  the Lie subalgebra of $u(N)$ which commutes with all the elements of $G$. Every possible dynamical Lie algebra ${\cal L}$ has to be a Lie subalgebra of $u^{G}(N)$. In appropriate coordinates, $u^{G}(N)$  takes a block diagonal form where each block is an arbitrary element of $u(\bar n)$ for appropriate dimension $\bar n$. The dynamical Lie algebra ${\cal L}$  also takes a block diagonal form and the underlying Hilbert space of the system splits  into the direct sum of invariant subspaces. If it is possible to perform any (special) unitary operation on each of the invariant subspaces, the system is called {\it subspace controllable}. Subspace controllability has recently  been the topic  of several studies  for several types of systems (see, e.g., \cite{confraJMP}, \cite{Marco}, \cite{Wang1}, \cite{Wang2}). The splitting of the Hilbert space is a consequence of the relation between the representations of $u^{G}(N)$ and the representations of the finite group $G$, a property often referred to as {\it Schur-Weyl duality}  (see  \cite{conJonas} for a review in the context of control theory). Bases for the invariant subspaces (or equivalently the change of coordinates that puts the system in block diagonal form) can be obtained from the  knowledge of the so-called {\it Generalized Young Symmetrizers (GYS)} which are projection matrices belonging to the {\it group algebra} (see, e.g., \cite{FH})  of the group $G$. In alternative,  one can circumvent the use and knowledge of GYS's by a purely Lie algebraic approach as described in section 4.3.4. of \cite{mikobook}.

The case where $G=S_n$, the permutation group on $n$ {objects, and such objects are two dimensional quantum system (${\it qubits}$),  is the most studied \cite{confraJMP}, \cite{Marco}. In these studies,  it is assumed that arbitrary local transformations, which coincide on any qubit,   are available along with one or more symmetric $k$-body Hamiltonians, the standard case being the one of a single 
$2-$body Hamiltonian.\footnote{There are also studies where symmetric interaction is allowed only between certain subsystems, in which case the symmetry group is a proper subgroup of $S_n$ \cite{confraSCL}.}  The results in \cite{Marco} show that, in such cases, the dynamical Lie algebra generated by the given Hamiltonians is $su(\dim(V))$ for any invariant subspace $V$,  plus the span of  elements which are  a multiple of the identity on each invariant subspace. However  the dynamical Lie algebra is, except for special cases,  not the full $u^{S_n} (2^n)$,\footnote{nor its subalgebra of trace zero matrices $su^{S_n}(2^n)$.} because it lacks control on the relative phases between the various invariant subspaces. This corrected a previous computational mistake in \cite{confraJMP} where, for the case of $2$-body interaction, it was claimed this to be the case. Subspace controllability however was claimed both  in \cite{confraJMP} and \cite{Marco} and the proof is similar,  based on induction and direct calculations of the Lie brackets. 

In this paper, with the goal of generalizing these results to systems of arbitrary dimension $d \geq 2$, {\it qudits}, and arbitrary Hamiltonians, we take a different route. The starting point is recognizing that the above splitting of the Hilbert space 
${\cal H} =(\mathbb{C}^d)^{\otimes n}$ coincides with the Clebsch-Gordan (CG)  (see, e.g., \cite{Alex}) splitting into irreducible representations of $su(d)$ for the tensor product of  $n$ irreducible {\it standard  representations}. The invariant subspaces are therefore (possibly repeated) {\it irreducible modules} of $su(d)$. The Lie algebra 
$u^{S_n}(d^n)$ is the direct sum of all the $su(\dim(V))$'s where $V$  runs over all the irreducible modules that appear in the CG decomposition,  plus its center. Such a center will  be characterized as the Abelian algebra (or Lie algebra) generated by the {\it Casimir operators} which act as scalars on each irreducible module with the value of the scalars depending on (and labeling) the various (nonisomorphic) irreducible modules. With this characterization of $u^{S_n}(d^n)$, in order to describe  the dynamical Lie algebra generated by a certain set of symmetric Hamiltonians, one splits such Hamiltonians according to their components on the center and on the orthogonal complement\footnote{which is the direct sum of the $su(\dim(V))$'s}. Under a semisimplicity assumption (verified in the case of subspace controllability) the dynamical Lie algebra will be the direct sum of the component generated by such orthogonal complements and the span of the component of the generators on the center. This is a general fact valid for any set of generating symmetric Hamiltonians. Focusing on the situation where the dynamical Lie algebra is generated by symmetric {\it local} operations and a two body Hamiltonian we can recast in this framework the results of \cite{confraJMP} and \cite{Marco} and prove subspace controllability for a new case, the case of {\it three qutrits}.

The paper is organized as follows: In section \ref{Background} we recall some general facts about irreducible representations ({\it irreps}) of $su(d)$.  In section \ref{sudn}, we describe the structure of $u^{S_n}\left( d^n\right)$ and establish the connection with the Clebsch-Gordan decomposition of $(\mathbb{C}^d)^{\otimes n}$. This section contains a detailed discussion of the {\it Casimir algebra}, that is the algebra generated by the Casimir operators which is the center of $u^{S_n}(d^n)$.  In section \ref{dynLA}, we describe the general structure of a dynamical Lie algebra generated by symmetric Hamiltonians, that is, Hamiltonians that are invariant under the action of the  
permutation group. Within this framework,  we reinterpret the results of \cite{Marco} for the case of networks of $n$ qubits in section \ref{caseofnqubits}. The new case of three qutrits is treated in section \ref{threequtrits} where we show subspace controllability for this system. This section also presents some results for the case of general $n$ qutrits. In section \ref{CandD} we draw some conclusions. In particular we discuss how the treatment for the case of three qutrits could be used as a blueprint to prove the general subspace controllability property for a system of $n$ qudits. The appendix contains a number of auxiliary and complementary results and the proofs omitted from the main text.

\section{Generalities about $su(d)$ representations}\label{Background}

We recall some known facts about irreducible representations of $su(d)$ that we will need in the paper emphasizing the {\it Gelfand-Tsetlin (GT) formalism}. We shall mostly follow \cite{Alex} to which we refer for further details and for references to the original literature.\footnote{The paper \cite{Alex} is  linked to a  useful web-site where one can carry out computations concerning irreducible representations of $su(d)$. }  However, there will be some small adjustments to prepare for the theory of the following sections.

\subsection{Labeling irreducible representations of $su(d)$}
 
 An {\it  irreducible  representation ({irrep})}  of $su(d)$ is identified by a $d-$tuple of natural number, the {\it i-weight},  
 $(m_{1,d},  m_{2,d},...,m_{d,d})$ with $m_{1,d} \geq m_{2,d} \geq \cdots \geq m_{d,d}$. $d-$tuples that differ by an integer in all entries correspond to the same irrep. Therefore it is customary to set $m_{d,d}=0$. The representation $S$  corresponding to   $(m_{1,d},  m_{2,d},...,m_{d,d})$ has dimension given by the formula 
 \begin{equation}\label{dimensione}
 \dim(S)=\prod_{1 \leq r < s \leq d} \left( 1+ \frac{m_{r,d}-m_{s,d}}{s-r}\right). 
 \end{equation}

There are other ways to label irreps of $su(d)$ related to the i-weight above. For example, one often considers the {\it quantum numbers} $p_1:=m_{1,d}-m_{2,d}$, 
$p_2:=m_{2,d}-m_{3,d}$,...,$p_{d-1}:=m_{d-1,d}-m_{d,d}$, or equivalently, the quantum numbers $\hat p_1:=m_{1,d}-m_{d,d}$, 
$\hat p_2:=m_{2,d}-m_{d,d}$,...,$\hat p_{d-1}:=m_{d-1,d}-m_{d,d}$. Another  popular way, 
 is to use {\it Young diagrams}, where the representation  $S:= \left( m_{1,d},  m_{2,d},...,m_{d,d} \right)$ corresponds to a Young diagram with $m_{1,d}$ boxes in the first row, $m_{2,d}$ boxes in the second row and so on. For example we have for $d=3$ 
$$
\left( 4,2,0 \right) \longleftrightarrow \young(\quad \quad \quad \quad ,\quad \quad ). 
$$

An alternative equivalent labeling scheme is through the so called {\bf Casimir operators}.(see, e.g., \cite{Pais} and the references therein). They are defined as follows. Consider a module $V$ for an  irreducible  representation of $su(d)$ and an isomorphism  $V \rightarrow V$ which commutes with all the elements of the representation. Then according to 
Schur Lemma (see, e.g., \cite{FH}) such an isomorphism has to be a multiple of the identity. For $su(d)$ there are $d-1$ independent such operators which are called Casimir operators, $C_2,C_3,...,C_d$.  
They are constructed \cite{CasConstr}, \cite{Bieden} starting from an orthonormal basis of $isu(d)$, $\{F_1,...,F_{d^2-1}\}$  and forming homogeneous quadratic for $C_2$, cubic   for $C_3$, and so on polynomials in these operators  with appropriate coefficients. The simplest and most famous case is the case of $su(2)$ for which the only Casimir operator is the quadratic one $C_2$ given by 
\begin{equation}\label{Casimirsu2}
C_2:=S_x^2+S_y^2+S_z^2,
\end{equation}  
where $S_{x,y,z}$ are the angular momentum operators in the $x,y,z$  directions (corresponding in the standard representation to the Pauli matrices $\sigma_x$, $\sigma_y$, and $\sigma_z$).  Any Casimir operator acts as a multiple of the identity on a given irreducible representation. Such a multiple is called the 
{\it Casimir eigenvalue} and it depends on the representation. Different (nonisomorphic)  irreducible representations may have the same value for one Casimir eigenvalue. However the set of Casimir eigenvalues for 
$C_2,C_3,...,C_d$ ($d-1$ numbers $c_2,...,c_{d-1}$) uniquely determines the irreducible representation.

\subsection{Bases of irreducible modules}
Given an irreducible representation of $su(d)$,  $S=\left(m_{1,d},...,m_{d,d} \right)$, elements of a basis are in one to one correspondence with the so-called {\it  Gelfand-Tsetlin patterns}  which are triangular patterns with $d$ rows $j=d,d-1,...,1$ made up of $j$ elements $(m_{1,j},m_{2,j},...,m_{j,j})$ which are natural numbers. The rows are listed in a way that the row $d$ is the upper most, followed by the row $d-1$, and so on up to the lowest row which contains a single element $m_{1,1}$. The first row is the same for every element of the basis  and coincides with the signature labeling the representation, $S:=(m_{1,d},  m_{2,d},...,m_{d,d})$, that is, its i-weight. The elements $m_{k,l}$, 
$ 1 \leq l \leq d$, $1 \leq k \leq l$ in a GT pattern have to satisfy the {\it betweenness conditions}, i.e.,\footnote{This means from any element go one step down and then one step up moving from left to right, you will have a nonincreasing sequence of three numbers. See the examples (\ref{A1S})) below.}  
\begin{equation}\label{betweenness1}
m_{k,l} \geq m_{k,l-1} \geq m_{{k+1},l}
.
\end{equation}
This restricts the number of possible patterns (states) to a number given by $\dim(S)$ in (\ref{dimensione}). For example for $d=3$ and the representation $S=(2,1,0)$, from (\ref{dimensione}) 
we have $\dim(S)=8$ and we have the following GT patterns-states. 

\begin{equation}\label{A1S}
A_1:=\begin{pmatrix} 2 & \quad &  1 & \quad & 0 \cr 
\quad & 2 & \quad & 1 & \quad  \cr 
\quad & \quad & 2 & \quad & \quad \end{pmatrix}= \young(00,1),  \qquad w(A_1)=(2,1,0), 
\end{equation}

$$
A_2:=\begin{pmatrix} 2 & \quad &  1 & \quad & 0 \cr 
\quad & 2 & \quad & 1 & \quad \cr 
\quad & \quad & 1 & \quad & \quad \end{pmatrix}= \young(01,1),  \qquad w(A_2)=(1,2,0), 
$$

$$
A_3:=\begin{pmatrix} 2 & \quad &  1 & \quad & 0 \cr 
\quad & 2 & \quad & 0 & \quad\cr 
\quad & \quad & 2 & \quad & \quad \end{pmatrix}= \young(00,2),  \qquad w(A_3)=(2,0,1), 
$$

$$
A_4:=\begin{pmatrix} 2 & \quad &  1 & \quad & 0 \cr 
\quad & 2 & \quad & 0 & \quad \cr 
\quad & \quad & 1 & \quad & \quad \end{pmatrix}= \young(01,2),  \qquad w(A_4)=(1,1,1), 
$$

$$
A_5:=\begin{pmatrix} 2 & \quad &  1 & \quad & 0 \cr 
\quad & 2 & \quad & 0 & \quad  \cr 
\quad & \quad & 0 & \quad & \quad \end{pmatrix}= \young(11,2),  \qquad w(A_5)=(0,2,1), 
$$

$$
A_6:=\begin{pmatrix} 2 & \quad &  1 & \quad & 0 \cr 
\quad & 1 & \quad & 1 & \quad  \cr 
\quad & \quad & 1 & \quad & \quad \end{pmatrix}= \young(02,1),  \qquad w(A_6)=(1,1,1), 
$$

$$
A_7:=\begin{pmatrix} 2 & \quad &  1 & \quad & 0 \cr 
\quad & 1 & \quad & 0 & \quad  \cr 
\quad & \quad & 1 & \quad & \quad \end{pmatrix}= \young(02,2),  \qquad w(A_7)=(1,0,2), 
$$

$$
A_8:=\begin{pmatrix} 2 & \quad &  1 & \quad & 0 \cr 
\quad & 1 & \quad & 0 & \quad  \cr 
\quad & \quad & 0 & \quad & \quad \end{pmatrix}= \young(12,2),  \qquad w(A_8)=(0,1,2), 
$$

In (\ref{A1S}), we have also denoted the (alternative) representation of the state in terms of {\it semistandard Young tableaux (SSYT)} which are Young diagrams (corresponding to the given irreps) filled with integer numbers in $\{0,1,...,d-1\}$,  in nondecreasing order row-wise and in strictly increasing order column-wise.\footnote{In view of the development that will follow we have replaced the standard notation which uses  numbers $1,...,d$ with $0,1,...,d-1$.} There is a simple algorithm to go from a GT pattern to the corresponding SSYT and viceversa (see,e.g., \cite{Alex}).\footnote{Assume we have the GT pattern. The top row of the pattern $(m_{1,d},...,m_{d,d})$ indicates the shape of the Young tableaux, for which the first row has $m_{1,d}$ boxes, the second row has $m_{2,d}$ boxes and so on. That said, roughly speaking, the diagonals of the GT pattern correspond to rows of the SSYT and the rows in the GT pattern correspond to the numbers that fill the various rows in the SSYT. More specifically assume the GT pattern is given. The number at the bottom indicates the number of $1$'s in the first row. Then the next number up in the first diagonal minus the number on the bottom indicates the number of $2$'s in the first row. Then the next number up minus the current number indicates the number of $3$'s and so on. Then we move to the next GT diagonal (SSYT row). The number on the bottom indicates the number of $2$'s. The next number up minus the current number indicates the number of $3$'s and so on.  For example consider the pattern $A_5$ in (\ref{A1S}). $0$ at the bottom of the GT pattern says that there is no $1$ in the first row of the SSYT. Then the next $2$ indicates $2-0$ $2$'s in the first row and the following $2$ indicates $2-2=0$ $3$'s in the first row. Then we go to the next GT diagonal (SSYT row). $0$ indicates that there is no $2$'s. $1-0=1$ indicates that there is one $3$.  This gives the associate SSYT of $\young(22,3)$ which coincides with the one indicated in (\ref{A1S} ) since we make the notational change $1 \rightarrow 0$, $2 \rightarrow 1$ $3 \rightarrow 2$.  The process can be inverted following the principle that  diagonals in the GT pattern correspond to rows in the SSYT and rows in the GT pattern correspond to numbers $1$, $2$ ..., from bottom to top. The number of occurrences of $1$ in the first row of the SSYT is the number at the bottom of the GT pattern. The number above it along the diagonal is the bottom number + the number of occurrences of $2$'s in the first row of the SSYT. Call the number $m$. The number above in the first diagonal is $m$+ the number of occurrences of $3$ in the first row of the SSYT, and so on to complete the first diagonal. Then we move to the second diagonal of the GT pattern and second row of the SSYT. The number at the bottom is the number of occurrences of $2$'s in the second row of the SSYT. Call this number $m$, $m$ + the number of occurrences of $3$ in the second for of the SSYT gives the next number up, and so on. For example consider $A_5$ in (\ref{A1S}) and assume the SSYT $\young(22,3)$ is given. Consider the first row of the SSYT. Since there are no $1$'s the number at the bottom of the first diagonal is $0$. This number + the number of occurrences of $2$'s gives the second number on the diagonal which is $2$. This number + the number of occurrences of $3$'s which is $0$ gives the top number which is $2$. Moving to the next GT diagonal (SSYT row) the number of $2$'s which is $0$ gives the bottom number. This number + the number of occurrences of $3$ which is $1$ gives the next number up which is $1$. }  To each basis element $A$ in (\ref{A1S}) it is   associated a {\it weight vector}  $w(A)$ which is  defined as $w(A)=(w_1,w_2,...,w_d)$ where $w_1$ is the number of $0$'s in the SSYT, $w_2$ is the number of $1$'s and so on up to $w_d$ which is the number of $d-1$'s in the SSYT, or equivalently, $w_l=\sigma_l-\sigma_{l-1}$ where $\sigma_j$ is the sum of the entries of the $j$-th row, with $\sigma_0:=0$.  Weight vectors will be very important for us because, as we shall see,  they are the main tool to translate the basis elements as written in the GT or SSYT formalism to basis in terms of the computational basis in $(\mathbb{C}^d)^{\otimes n}$. It is also important to observe that weight vectors are {\it not} in one to one correspondence with basis elements. This happens only in the case of $su(2)$. We shall see, in the next subsection, that states with the same weight vector are eigenvectors of operators $S_z^l$, $l=1,...,d-1$ corresponding to degenerate eigenvalues (cf. (\ref{eigeig})).

Of particular interest to us is the {\it standard defining representation} which is represented by the i-weight $(1,0,0,...,0)$ and has dimension $d$.\footnote{We will take repeated tensor products of such representations.} This is the representation one normally write down when writing matrices in $su(d)$.  It is spanned  by patterns where only the first diagonal is different from zero. For example, in the case $d=3$, it is spanned by the states 
$$
S_1:=\begin{pmatrix}  
1        & \quad &      0    & \quad &   0 \cr 
\quad &      1   & \quad  &     0    &   \quad \cr
\quad &   \quad &    1     &    \quad & \quad 
\end{pmatrix}, \quad S_2=\begin{pmatrix}  
1        & \quad &      0    & \quad &   0 \cr 
\quad &      1   & \quad  &     0    &   \quad \cr
\quad &   \quad &    0     &    \quad & \quad 
\end{pmatrix}, \quad 
S_3=\begin{pmatrix}  
1        & \quad &      0    & \quad &   0 \cr 
\quad &      0   & \quad  &     0    &   \quad \cr
\quad &   \quad &    0     &    \quad & \quad 
\end{pmatrix}. 
$$ 
States of the standard representation $\{S_1,S_2,...,S_d\}$ are identified by their weight vectors, and we have $w(S_1)=(1,0,...,0)$, $w(S_{2})=(0,1,0,...,0)$,...,$w(S_d)=(0,0,...,0,1)$. This representation corresponds to a Young diagram with only one box, and the basis states are SSYT with the box filled by $0,1,...,d-1$.

\subsection{Operators $S_z^l$, and raising and lowering operators $S_{\pm}^l$, $l=1,...,d-1$} \label{duetre}
To describe a given representation one considers  the {\it operators $S_z^l$} and the {\it raising and lowering operators  $S_{\pm}^l$, $l=1,...,d-1$}. The operators $S_z^l$ are the operators such that every basis state $M$  is an eigenvector  of $S_z^l$ with eigenvalue given by $\frac{w_l(M)-w_{l+1}(M)}{2}$, where $w(M)$ denotes the weight vector of $M$, that is,   
\begin{equation}\label{eigeig}
S_z^l(M) =\frac{w_l(M)-w_{l+1}(M)}{2} M.
\end{equation}
 for $l=1,2,...,d-1$ (cf., e.g., section IV in \cite{Alex}). Notice that there is a one to one correspondence between the weight vectors and the eigenvalues of the operators $S_z^l$, $l=1,...,d-1$, that is, the weight vectors  determine the eigenvalues according to (\ref{eigeig}) and viceversa,  given the eigenvalues $\lambda_1(M),...,\lambda_{d-1}(M)$, they uniquely determine the weight vector.\footnote{In fact the system $w_l-w_{l+1}=2 \lambda_l $, $l=1,...,d-1$, $\sum_{l=1}^d w_l=\sigma_d$ (which is always verified.)  has a unique solution (the matrix of coefficients of the system of the $d$ equations above is easily seen by induction on $d$ to have determinant $d \not=0$).}

 The raising $S_+^l$ and lowering $S_-^l$ operator, $l=1,...,d-1$ act on a state $M$ by giving a linear combination of states. Fix $l\in \{ 1,...,d-1\}$.  The states that appear in the linear combination are the ones obtained by adding, for $S_+^l$, or subtracting, for $S_-^l$,  $1$  to one element of the $l$-th row of the GT pattern and eliminating  the patterns that are not feasible (that is,  do not satisfy the betweenness condition (\ref{betweenness1})).  The coefficients for such linear combinations are known explicitly (see formula (28) and (29) in \cite{Alex}). A special state called the {\it highest weight state} $\hat H$  of the representation is the unique state (GT pattern) such that each diagonal is constant (such as in $A_1$ in (\ref{A1S})). This  state $\hat H$ is such that the weight vector $w(\hat H)$ coincides with the i-weight labeling  the representation and it has the property that $S_+^l (\hat H)=0$, $\forall \, l=1,2,...,d-1$. The weight vector of any state appearing in $S_\pm^l(M)$ is 
 $\left( w_1,w_2,...,w_{l-1},w_l\pm 1, w_{l+1}\mp 1,w_{l+2},...,w_d \right)$ if $w(M)=\left( w_1,w_2,...,w_{l-1},w_l, w_{l+1},w_{l+2},...,w_d \right)$. The next two observations will be important for what follows.
 \begin{itemize}
 \item {\bf Fact 1} {\it The space ${\cal S}_j$ defined as the subspace spanned by the states $M$ with $w_d(M)=j$ for fixed $j$ is invariant under $S_z^l$, $S_{\pm}^l$, $l=1,...,d-2$. }
 
 \item {\bf Fact 2} {\it  For $l=1,2,...,d-2$, all the states with a given $d-1$ row 
 $\left(m_{1,d-1}, m_{2,d-1},...,m_{d-1,d-1} \right)$ together 
 with $S_{\pm}^l$, $S_z^l$, $l=1,...,d-2$ give the representation of $su(d-1)$ with  i-weight 
 $S= \left(m_{1,d-1}, m_{2,d-1},...,m_{d-1,d-1} \right)$.}
  \end{itemize}
  
The `modularity property' of {\bf Fact 2} says that representations of $su(d-1)$ are somehow embedded in representations of $su(d)$, a fact that is already apparent  if one looks at the standard representation. This fact  follows by examining the coefficients in the recalled formulas (28) and (29) of \cite{Alex}, where the coefficients involved in $S_{\pm}^l(M)$ only depend on the entries of $M$ up to row $l+1$.

\subsection{Tensor products of representations and Clebsch-Gordan decomposition}

When considering the tensor product of two irreducible representations 
$S \otimes S^{'}$, the resulting representation  is, in general, not irreducible but it is   the direct sum of certain irreducible representations which may appear with various multiplicities. To understand which irreps appear in the tensor product and with which multiplicity, one can use a variation of the Littlewood-Richardson rule as explained in \cite{Alex}:  

{\bf Algorithm 1} {\it Let $S \otimes S^{'}$ the tensor product representation of representations $S$ and $S^{'}$.  Take a basis of $S^{'}$ (the result would be the same by using  $S$ instead of  $S^{'}$) and to each pattern associate the corresponding $B$-pattern obtained replacing $m_{k,l}$, for $l=1,...,d$, $k=1,...,l$ with $b_{k,l}:=m_{k,l}-m_{k,{l-1}}$, by setting $m_{k,0}:=0$. For example, the $B$-pattern associated with $A_1$ in (\ref{A1S})  is  (notice that $B$-patterns do not need to   
satisfy the betweenness condition (\ref{betweenness1})) 
$$
B_1=\begin{pmatrix}  
0        & \quad &      0    & \quad &   0 \cr 
\quad &      0   & \quad  &     1    &   \quad \cr
\quad &   \quad &    2     &    \quad & \quad 
\end{pmatrix}, 
$$ 
Then consider each of these $B$-patterns and take the i-weight  of the other irrep, that is,  $S:=(m_{1,d}, m_{2,d},...,m_{d,d})$. Follow the diagonals of the $B$-pattern from left to right and its elements from top  to bottom. For  each diagonal, add each   element to the signature starting from the last element and proceeding towards left. If at each step the resulting i-weight is a legit i-weight, i.e., 
$m_{1,d} \geq m_{2,d} \geq \cdots \geq m_{d,d}$, the final i-weight is an i-weight  of an 
irrep appearing in $S^{'} \otimes S$.

 In more explicit  terms,  the algorithm is as follows. Starts with $(m_{1,d},...,m_{d,d})$ (of $S$) and add $(0,0,...,0,b_{1,d})$. If the result is legit, add $(0,0,...,b_{1,d-1},0)$.  If the result is legit, add $(0,0,...,b_{1,d-2},0,0)$, and so on up to adding $(b_{1,1},0,...,0)$. Then continue considering the second diagonal by adding $(0,0,...,0,b_{2,d})$ and if the result is legit add $(0,0,...,b_{2,d-1},0)$, and so on up to $(0,b_{2,2},0,...,0)$, and so on for all diagonals until the last one which is trivial (we add $(0,0,...,0)$ because $b_{d,d}$ is always zero since $m_{d,d}$ is normalized to be zero).

Each time an i-weight   results from such a process accounts for one time the corresponding representation appears in the tensor product. Therefore,  this method also allows to find the multiplicities of the representations.  
Consider for example the $B_1$ above and assume $S=(3,2,0)$. The process goes as follows 
$$
(3,2,0)+(0,0,0) \rightarrow (3,2,0)+ (0,0,0)\rightarrow (3,2,0)+ (2,0,0) \rightarrow (5,2,0) + (0,0,0) \rightarrow 
$$
$$
 (5,2,0) + (0,1,0) \rightarrow   (5,3,0) + (0,0,0)=  (5,3,0). 
$$}

\qed

\vspace{0.5cm}

If $S^{''}$ appears in the decomposition of $S \otimes S^{''}$, a basis state $M^{''}$ can be written as a linear combination 
$$
M^{''}=\sum_{M,M^{'}} C_{M,M^{'}}^{M^{''}} M \otimes M^{'}, 
$$
where $M$ and $M^{'}$ run over the bases of $S$ and $S^{'}$ respectively and the coefficients $C_{M,M^{'}}^{M^{''}}$ are the  {\it Clebsch-Gordan coefficients} to which much literature is devoted. They satisfy the so-called {\it selection rules}, concerning the weights $w$ of each of the states. These impose that $w(M^{''})\not=w(M)+w(M^{'}) \Rightarrow   C_{M,M^{'}}^{M^{''}} =0$, where the sum of the weights is defined to be  component-wise. 

In the case one considers the tensor products of more than two representations, the previous setting is applied several times by using the distributive property of the tensor product of representations with respect to the direct sum, that is, 
for irreducible modules $S_1,...,S_m$ and $\hat S$, we have 
$$
\left( S_1 \oplus S_2 \oplus \cdots \oplus S_m \right) \otimes \hat S=
 ( S_1 \otimes \hat S) \oplus  ( S_2 \otimes \hat S) \oplus \cdots \oplus ( S_m \otimes \hat S) 
$$

\section{The Lie algebra $u^{S_n}(d^n)$ and the Clebsch-Gordan decomposition of $(\mathbb{C}^d)^{\otimes n}$}\label{sudn}

\subsection{Clebsch-Gordan decomposition for tensor products of standard representations} 

When dealing with quantum information problems, the basis of choice for  the irreducible module 
$\mathbb{C}^d$ (for the standard representation) is spanned by $\{ |0\rangle, |1\rangle,...,|d-1 \rangle$ and the correspondence with the SSYT and (GT)  weight vector is the obvious  one\footnote{For the standard representation the correspondence between GT pattern states  and weight vectors is one to one.}  
$$\young(j) \leftrightarrow |j\rangle \leftrightarrow w=(0,0,..,0,1,0,..,0), $$ 
where the $1$ in the weight vector $w$ appears in the $(j+1)$-th position.  The operators $S_z^l$, $S_{\pm}^l$, for $l=1,...,d-1$ are defined as\footnote{Notice that to keep track of standard quantum information theory notation we made a shift of indexes where the state $|l-1\rangle$ represents the state with weight $w_l=1$ and all other weights equal to zero.} 
\begin{equation}\label{Szpmstandard}
S_z^l |l-1 \rangle =\frac{1}{2} |l-1 \rangle, \quad 
S_z^l |l \rangle =- \frac{1}{2} |l\rangle, \qquad S_+^l |l\rangle=
|l-1\rangle, \quad  S_-^l |l-1\rangle =|l \rangle, 
\end{equation}
with all the other applications of $S_{z,\pm}^l$ being zero.  In particular, notice that $iS_{\pm}^l$, $l=1,...,d-1$ are {\it not} in $su(d)$ but 
\begin{equation}\label{rrt}
iS_x^l:=i\frac{S_{+}^l+S_{-}^l}{2}, \qquad iS_y^l:=\frac{S_+^l -S_-^l}{2}
\end{equation} 
are.

When considering systems of {\it $n$ qudits},  we have  the Hilbert space ${\cal H}:=(\mathbb{C}^d)^{\otimes n}$ as the tensor product of $n$ standard representations.  In this case,  the operators $S_{z,\pm}^l$ are replaced by the corresponding 
\begin{equation}\label{totalOperators} 
\hat S_{z,\pm}^l:=S_{z,\pm}^l \otimes {\bf 1} \otimes \cdots \otimes {\bf 1}+ 
{\bf 1} \otimes S_{z,\pm}^l \otimes {\bf 1}  \cdots \otimes {\bf 1}+\cdots {\bf 1}\otimes \cdots \otimes {\bf 1} \otimes S_{z,\pm}^l. 
\end{equation}
This gives the {\it tensor product representation}  of $su(d)$ which  is not reducible and it splits into irreducible representations according to the {\it Clebsch-Gordan decomposition}, recursively. 

Let us examine which representations appear in this decomposition. Consider Algorithm 1 with a representation $\tilde S \otimes S^{'}$ where $S^{'}$ is the standard representation. For $S^{'}$ the standard representation, the $B$-patterns are particularly simple: only the first diagonal is different from zero and contains only one $1$'s and all zeros. In particular $B_{j}$, $j=1,...,d$,  is such that  $b_{1,d+1-j}=1$ and all the other entries  are equal to zero. If $\tilde S$ is a representation with i-weight  $S=(m_{1,d},m_{2,d},...,m_{d,d})$,  the representations  that appear in the tensor product with $S^{'}$, the standard representation, are the ones obtained from $S$ adding $1$ to one of the entries  in the i-weight  without violating  the rule $m_{1,d} \geq m_{2,d} \geq \cdots \geq m_{d,d}$. In particular every representation has necessarily multiplicity one in this case. We are interested in the tensor product of $n$ irreducible standard representations $S:=(1,0,...,0)$.  Therefore this procedure has to be iterated $n-1$ times. As a result, the same representation may appear more than once in the final (CG) decomposition. The following proposition clarifies the situation.

\begin{proposition}\label{clarif}
Let $S$ be the standard defining representation of $su(d)$. Then 
$$
S^{\otimes n}=\bigoplus_{ m_1+m_2+\cdots m_d=n} (m_1,m_2,...,m_d)^{\oplus k_{(m_1,m_2,...,m_d)}}, 
$$
where the sum runs over all the irreps of  $su(d)$ with i-weight  
$(m_1,m_2,...,m_d)$ such that $m_1+m_2+\cdots +m_d=n$,\footnote{Here we do not use the normalization convention where we set $m_d=0$, so that, for example, for $n=6$ $(2,2,2)$ is a legitimate representation  of $su(3)$ coinciding with $(0,0,0)$.}  and   $k_{(m_1,m_2,...,m_d)}$ denotes the {\it multiplicity} of $(m_1,m_2,...,m_d)$. The multiplicity $k_{(m_1,m_2,...,m_d)}$ is calculated recursively on $n$ and it is equal to  
\begin{equation}\label{multiplicity3}
k_{(m_1,m_2,...,m_d)}=k_{(m_1-1,m_2,...,m_d)}+k_{(m_1,m_2-1,...,m_d)}+\cdots+k_{(m_1,m_2,...,m_d-1)}, 
\end{equation}
where $k_{(j_1,...,j_d)}$ is set equal to zero for a non admissible $d-$tuple $(j_1,...,j_d)$, and $k_{(1,0,...0)}=1$.
\end{proposition}
\begin{proof}
We use induction on $n$.  First we show that {\it every}  
$(m_1,m_2,...,m_d)$ representation with $m_1+m_2+\cdots m_d=n$ is admissible.  If we consider $(m_1,m_2,...,m_d)$ and starting from $m_d$ and moving leftwards we choose the first index $m_j$ which is strictly positive so that $(m_1,m_2,...,m_j-1,0,...,0)$ is a possible pattern in $S^{\otimes n-1}$ we know, by the inductive assumption, to be included  in $S^{\otimes n-1}$. Using the $B_j$, $B$-pattern of the standard representation we obtain $(m_1,m_2,...,m_d)$. Since $(m_1,m_2,...,m_d)$ is arbitrary, this proves that {\it every}  representation  $(m_1,m_2,...,m_d)$, with $m_1+m_2+\cdots + m_d=n$ is present in $S^{\otimes n}$. Such a representation is obtained as many times as it is possible to find $(j_1,j_2,...,j_d)$ in 
$S^{\otimes n-1}$ such that by adding $(0,0,..,0,1,0,...,0)$ according to which $B$-pattern is used one obtains $(m_1,m_2,...,m_d)$. This is given by  
equation (\ref{multiplicity3}). 

\end{proof}


\begin{remark}\label{quantumnumbers}
If we use the quantum 
numbers $p_j:=m_j-m_{j+1}$, $j=1,...,d-1$, to identify a representation  $(m_1,...,m_d)$, the representations that appear  in the decomposition are all the ones corresponding to $(p_1,...,p_{d-1})$ with $p_j \geq 0$, $j=1,..,d-1$,  
$p_1=m-\sum_{j=2}^{d-1} j p_j $, 
for $m=n,n-d,n-2d,...,n-\lfloor \frac{n}{d}\rfloor d$.\footnote{Consider the representation $(m_1,...,m_d)$, according to the value of $m_d$ which can be $0,1,...,\lfloor\frac{n}{d}\rfloor$ we can normalize the signature by subtracting $m_d$ to each index. The sum of all the remaining   therefore becomes  $m=n,n-d,n-2d,...,n-\lfloor \frac{n}{d}\rfloor d$, according to whether $m_d=0,1,...,\lfloor \frac{n}{d}\rfloor$. Now fix $m=n,n-d,n-2d,...,n-\lfloor \frac{n}{d}\rfloor d$.  Inverting $p_j:=m_j-m_{j+1}\geq 0$, $j=1,...,d-1$, we obtain $m_j=\sum_{k=j}^{d-1} p_k$. This together with $\sum_{j=1}^{d-1} {m_j}=m$, gives $p_1=m-\sum_{j=2}^{d-1} j p_j$.} 

We shall consider  in particular, in the following,  the cases of $su(2)$ and $su(3)$. For $su(2)$,  the admissible representations occurring in $(\mathbb{C}^2)^{\otimes n}$  are, specializing the above formula, parametrized by a single parameter $p$ with 
\begin{equation}\label{ponly}
p=n,n-2,n-4,...,n-2\lfloor \frac{n}{2}\rfloor. 
\end{equation}
For $su(3)$,  we can write the admissible representations occurring in $(\mathbb{C}^3)^{\otimes n}$ in terms of quantum numbers as 
\begin{equation}\label{admissrep}
(m-2j,j),  \qquad  \texttt{with} \qquad  \, m=n,n-3,n-6,...,n-3 \lfloor\frac{n}{3}\rfloor \qquad j=0,1,...,\lfloor \frac{m}{2} \rfloor.  
\end{equation}
\end{remark}

The bases of the resulting  irreducible modules can be found by using formulas for the Clebsch-Gordan coefficients recursively. Such explicit formulas  are known for the case of the tensor product of two representations $S \otimes S^{'}$ when  one of the representations is the standard one \cite{Vilekin}, which is the situation of interest here.

\subsection{Relation between the Lie algebra $u^{S_n}(d^n)$ and the Clebsch-Gordan decomposition of 
$(\mathbb{C}^d)^{\otimes n}$}

It must be emphasized that the CG deconposition of $(\mathbb{C}^d)^{\otimes n}$ is unique only up to isomorphisms of the corresponding sub-representations, in the following sense: The matrices $i\hat S_{z,\pm}^l$ in (\ref{totalOperators})  form a representation of $su(d)$ with reducible module $(\mathbb{C}^d)^{\otimes n}$, which is by definition the tensor product representation.   Such a representation  is reducible into irreducible modules according to the Clebsch-Gordan decomposition. Consider now two decompositions of the module $(\mathbb{C}^d)^{\otimes n}$ into irreducible modules.  if we fix a certain representation $R$ and consider in decomposition 1 the sum of all the representations isomorphic to $R$, $\oplus_{j=1}^{\bar n_1} S_j$ and in decomposition 2 the sum of all the representations isomorphic to $R$, $\oplus_{j=1}^{\bar n_2} T_j$, we have $\bar n_1=\bar n_2$ and $\oplus_{j=1}^{\bar n_1} S_j=\oplus_{j=1}^{\bar n_1} T_j$. In particular, 
the irreducible representations that appear (up to isomorphism)
only once in the decomposition are uniquely determined. This is a standard fact in representation theory (cf. e.g., Theorem 4.2.7 in \cite{mikobook}). We shall in the following refer to any (isomorphic) decomposition of 
$(\mathbb{C}^d)^{\otimes n}$, reducible module of the tensor product representation of $su(d)$,  as a {\it Clebsch-Gordan decomposition (of $(\mathbb{C}^d)^{\otimes n}$}
).

The vector space $(\mathbb{C}^d)^{\otimes n}$ is also a module of $S_n$ (the permutation group of $n$ objects) and for $u(d^n)$ (the Lie algebra of skew-Hermitian $d^n \times d^n$ matrices; in this case it would be the standard representation).\footnote{Recall that we denote by $u^{S_n}(d^n)$ ($su^{S_n}(d^n)$) the subalgebra of $u(d^n)$ ($su(d^n)$) of matrices in $u(d^n)$ commuting with (the given representation of) $S_n$.}  If $\{F_0,F_1, ....,F_{d^2}\}$ is an orthonormal  basis of $iu(d)$,  
 with $\{F_0,F_1,....,F_{d^2} \}$ Hermitian matrices, setting $F_0:=\frac{1}{\sqrt{d}} {\bf 1}$, then a 
 basis of  $u^{S_n}(d^n)$ is obtained as follows.  Consider a $d^2$-tuple of natural numbers $\left( j_0,j_1,...,j_{d^2-1} \right)$ with 
 $j_0+j_1+\cdots + j_{d^2-1}=n$ and denote by $F_{(j_0,j_1,...,j_{d^2-1})}$ the Hermitian matrix which is the sum of all tensor products with $j_0$ $F_0$'s,  
 $j_1$, $F_1$'s,...,$j_{d^2-1}$, $F_{d^2-1}$'s. These sums\footnote{Each of these sums contains $\frac{n!}{j_0!j_d!\cdots j_{d^2-1}!}$ tensor products.}  form a basis  for $iu^{S_n}(d^n)$.   The dimension of $u^{S_n}(d^n)$ is equal to the number of all the possible $d^2-$tuples, 
 $(j_0,...,j_{d^2-1})$ with $j_0+j_1+\cdots+j_{d^2-1}=n$, which is equal to $\begin{pmatrix} n+ d^2-1 \cr d^2-1 \end{pmatrix}$.

  We remark that 
 $$
 su(d)^{\otimes n} \subseteq su^{S_n}(d^n), \qquad  u(d)^{\otimes n} \subseteq u^{S_n}(d^n). 
 $$
 In particular $u(d)^{\otimes n}$ ($su(d)^{\otimes n}$) is the Lie subalgebra of $u^{S_n}(d^n)$ ($su^{S_n}(d^n)$) spanned by elements 
 $iF_{(j_0,j_1,...,j_{d^2-1})}$, where $F_{(j_0,j_1,...,j_{d^2-1})}$ is of the form  $F_{(n,0,...,0)}$,   $F_{(n-1,1,0,...,0)}$,  
 $F_{(n-1,0,1,0...,0)}$,...,  $F_{(n-1,0,0,...,0,1)}$ (or $F_{(n-1,1,0,...,0)}$,  $F_{(n-1,0,1,0...,0)}$,...,  $F_{(n-1,0,0,...,0,1)}$, only, respectively). It is a known fact and it follows from Schur-Weyl duality (see, e.g., \cite{conJonas}) that, in appropriate coordinates on $(\mathbb{C}^d)^{\otimes n}$, the Lie algebra 
 $u^{S_n}(d^n)$ gets a block diagonal form where each block can be an 
 arbitrary matrix in $u(\bar n)$ for certain $\bar n$, that is, we have 
 \begin{equation}\label{directsum2}
 u^{S_n}(d^n)=\bigoplus_j u\left( n_j \right).
 \end{equation}
 Therefore every (dynamical) Lie subalgebra of $u^{S_n}(d^n)$ also takes such a block diagonal form. The irreducible invariant vector spaces $V_j$ form a decomposition of 
 $(\mathbb{C}^d)^{\otimes n}:=\oplus_j V_j$, which, again, is not unique (see, for a discussion and a method to find {\it all} the decompositions Chapter 4 of \cite{mikobook} and \cite{confraQIC}). The crucial point of our approach is that such decompositions coincide with the Clebsch-Gordan decompositions. We have the following.

  \begin{theorem}\label{alwaysassumed}
 Consider a decomposition $(\mathbb{C}^d)^{\otimes n}=\oplus_l V_l$. It is a (Clebsch-Gordan) decomposition into irreducible representations for $su(d)^{\otimes n}$ ($u(d)^{\otimes n}$) if and only  it is an irreducible decomposition for $su^{S_n}(d^n)$ ($u^{S_n}(d^n))$). 
 \end{theorem}
 This fact was used (without a formal proof) in \cite{confraSCL} for the case $d=2$. We give here a general proof for any $d \geq 2$.  To do that we have to elaborate a little bit more on the basis $\{ F_{(j_0,....,j_{d^2-1})} \}$ ($j_0+j_1+\cdots + j_{d^2-1}=n$) we have considered for $iu^{S_n}(d^n)$ recalling that the subset of matrices where $j_0=n$ and $j_0=n-1$ gives a basis of 
 $iu(d)^{\otimes n}$. The original basis of $iu(d)$ we have considered, that is,  $\{F_0=\frac{1}{\sqrt{d}}{\bf 1}, F_1,...,F_{d^2-1}\}$ is also a basis of $gl_C (d)$ as a complex Lie algebra (the complexification of $u(d)$) and for $gl_C(d)$ we can define (analogously to what we have done for $u(d)$) the complex Lie algebras $gl_C(d)^{\otimes n}$ and $gl_C^{S_n}(d^n)$.  Complex  representations of $u(d)^{\otimes n}$ coincide with representations of $gl_C(d)^{\otimes n}$ while complex representations of $u^{S_n}(d^n)$ coincide with representations of 
 $gl_C^{S_n}(d^n)$.\footnote{This is the known fact from 
 representation theory (see, e.g, \cite{FH}): Complex representations of a Lie algebra are in one to one correspondence with complex representations of its complexification.} Furthermore  $\{ F_{(j_0,....,j_{d^2-1})} \}$ ($j_0+j_1+\cdots + j_{d^2-1}=n$) also gives a basis of $gl_C^{S_n}(d^n)$. Define in $gl_C^{S_n}(d^n)$ the 
  following subspaces for $j=0,1,...,n$
$$
 {\cal F}_j=\texttt{span}_{j_0=n-j} \{ F_{(j_0,j_1,...,j_{d^2-1})} \},   
$$
that is, $j$ counts how many positions in the tensor products are different from the identity. 
By definition,  ${\cal F}_0 \oplus {\cal F}_1=gl_C(d)^{\otimes n}$. Furthermore,  assume $B \in {\cal F}_j$, for $j=1,2,...,n-1$. Then 
  \begin{equation}\label{inclus3}
  AB \in {\cal F}_j, \qquad \texttt{for }{A \in {\cal F}_0}, \qquad \qquad   AB \in {\cal F}_{j-1} \oplus {\cal F}_{j} \oplus {\cal F}_{j+1} , \qquad \texttt{for }{A \in {\cal F}_1}
  \end{equation}
  \begin{proof}
According to  the above discussion, we shall equivalently prove that $(\mathbb{C}^d)^{\otimes n}=\oplus_l V_l$ is an irreducible 
 Clebsch-Gordan decomposition of $gl_C(d)^{\otimes n}$ if and only if it is an irreducible decomposition for $gl_C^{S_n} (d^n)$.

 Assume first that $(\mathbb{C}^d)^{\otimes n}=\oplus_l V_l$ is an irreducible decomposition of  $gl_C^{S_n} (d^n)$. Since $gl_C(d)^{\otimes n} \subseteq gl_C^{S_n} (d^n)$,  each subspace $V_l$ is invariant 
 for $gl_C(d)^{\otimes n}$, that is, it is a module for $gl_C(d)^{\otimes n}$ (or equivalently $gl_C(d)$). Since it is known  that every representation of $gl_C(d)$ is completely reducible, we can write $V_l$ as $V_l=V_l^1 \oplus V_l^2$ with $V_l^1$ irreducible (and $V_l^2$ invariant). We want to show that $V_l^1$ is invariant for $gl_C^{S_n} (d^n)$ also, which will imply $V_l^1=V_l$ and therefore $V_l$ irreducible for $gl_C(d)^{\otimes n} $. Since this will hold for general $l$ each $V_l$ is irreducible and therefore $(\mathbb{C}^d)^{\otimes n}=\oplus_l V_l$  is a Clebsch-Gordan  decomposition for $gl_C(d)^{\otimes n}$, which will conclude the proof of this direction of the theorem. To show that $V_l^1=V_l$, since $gl_C^{S_n}(d^n)=\oplus_{j=0,1,...,n} {\cal F}_j$, for $\vec v \in V_l^1$, we prove by induction on $j$ that ${\cal F}_j \vec v \in V_l^1$.\footnote{Recall that $j=0,1,...,n$ counts how many places in the tensor products appearing in $F_{(j_0,...,j_{d^2-1})}$ are not equal to the identity.} For $j=0$ and $j=1$ this is true since $V_l^1$ is invariant for $gl_C(d)^{\otimes n}$ and ${\cal F}_0 \oplus {\cal F}_1=gl_C(d)^{\otimes n}$. Assume now the claim true for up to $j$ and pick a basis element $F_{(j_0,...,j_{d^2-1})}$ which is in ${\cal F}_{j+1}$. For any  
$r=1,...,d^2-1$, we have $
F_{(n-1,0,...,0,1,0,...,0)}F_{(j_0+1,j_1,...,j_{r-1}, j_r-1, j_{r+1},....,j_{d^2-1})}=F_{(j_0,j_1,...,j_r,...,j_{d^2-1})}+B, 
$
with $B \in {\cal F}_{j-1} \oplus {\cal F}_j$ (cf. (\ref{inclus3})). From the inductive assumption it follows, since 
$$F_{(n-1,0,...,0,1,0,...,0)}F_{(j_0+1,j_1,...,j_{r-1}, j_r-1, j_{r+1},....,j_{d^2-1})}\vec v \in V_l^1
$$  
that $F_{(j_0,...,j_{d^2-1})} \vec v \in V_l^1.$ This proves the claim.

To prove the reverse implication,  assume that $(\mathbb{C}^d)^{\otimes n}=\oplus_l V_l$ is an irreducible Clebsch-Gordan decomposition of $gl_{C}(d)^{\otimes n}$. Then, as proven 
above by induction on $j$, every $V_l$ is invariant for $gl_C^{S_n}(d^n)$, that is,  a module for $gl_C^{S_n}(d^n)$.\footnote{The representation of  $gl_C^{S_n}(d^n)$ is also completely reducible since it coincides with the representation of $u^{S_n}(d^n)$ which is a unitary representation and as such completely reducible \cite{FH}.} Fix now one $V_l$, and write $V_l$ as 
$V_l=V_l^1 \oplus V_l^2$ with $V_l^1$ irreducible. Since $V_l^1$ is   invariant under  $gl_C^{S_n}(d^n)$, it is in particular also invariant under $gl_C(d)^{\otimes n}$, which contradicts the irreducibility of $V_l$, unless $V_l^1=V_l$.  Thus $V_l$ is irreducible for $gl_C^{S_n}(d^n)$ as well.

 \end{proof}

 \subsection{Casimir operators and the center of $u^{S_n}(d^n)$.}
 
 According to Theorem  \ref{alwaysassumed},  in appropriate coordinates,    
$(\mathbb{C}^d)^{\otimes n}$ splits in  a Clebsch-Gordan direct sum of irreducible representations  of $su(d)$ which are invariant and irreducible  under $u^{S_n}(d^n)$. On each of them, any Casimir operator acts as a multiple of the identity, where the value of the `multiple' depends on the representation.

Expressions for the Casimir operators for any representation of $su(d)$  were given for example in \cite{Pais}. Let us specialize these expressions to the sub-representations of $u(d)^{\otimes n}$.

Given ${\cal B}:=\{ F_0=\frac{1}{\sqrt{d}}{\bf 1}, F_1 ,..., F_{d^2-1}\}$ an orthonormal basis of $iu(d)$ (in the standard representation), as defined above, we extend notation (\ref{totalOperators}) by saying 
$\hat F_j:=F_{n-1,0,0,...,0,1,0,...,0}$ where the $1$ appears in position $j$, that is the sum of all tensor products  having only one $F_j$ and identities everywhere else. 
The Casimir operators on $(\mathbb{C}^d)^{\otimes n}$ are defined in terms of these matrices. In particular,  we have (cf., e.g.,  formula (3.59) in \cite{Pais} and cf.  \cite{Bieden}) that the Casimir operators of order $2,...,d$, $C_2$, $C_3$,...,$C_d$ are (all indexes in the sums run from $1$ to $d^2-1$)

\begin{equation}\label{Casimirsu2plus}
 C_2=\sum_{l} \hat F_l^2,  
\end{equation}

$$ C_3=\sum_{l,m,q} d_{l,m}^q \hat F_l \hat F_m \hat F_q,  $$
$$ C_4=\sum_{l,m,r,s} \left( \sum_{q} d_{l,m}^q d_{q,r}^s\right) \hat F_l \hat F_m \hat F_r \hat F_s,  $$
$$C_5=\sum_{l,m,r,s,t} \left( \sum_{a,b} d_{l,m}^a d_{a,b}^rd_{b,s}^t\right) \hat F_l \hat F_m \hat F_r \hat F_s \hat F_t,  $$
$$\vdots $$
$$C_{d}=\sum_{l_1,l_2,...,l_d}  \left( \sum_{a_1,a_2,...,a_{d-2}} d_{l_1,l_2}^{a_1} d_{a_1, a_2}^{l_3} d_{a_2, a_3}^{l_4} \cdots d_{a_{d-2}, l_{d-1}}^{l_d}\right)   \hat F_{l_1} \hat F_{l_2}\cdots \hat F_{l_d}. $$
where $d_{j,k}^l$ is a {\it symmetric} $3-$ tensor which defines the 
anti-commutation  relations of the elements in the basis ${\cal B}$.\footnote{That is $\{ F_j, F_k \}:=FjF_k+F_k F_j=\gamma {\bf 1}+ \sum_{l} d_{j,k}^l F_l$, for some scalar $\gamma$.} From the symmetry of $d_{j,k}^l$,  it follows that $C_{2,3,...,d}$ also belong  to $iu^{S_n} (d^n)$ and so does any power $C_{2,3,...,d}^k$, for $k=0,1,2,...$, 
and, more in general, any product of the Casimir operators and therefore the commutative algebra generated by the Casimir operators.\footnote{Notice that such an algebra is commutative since the space 
$(\mathbb{C}^d)^{\otimes n}$ splits in the direct sum of irreps on which each Casimir operator acts as a multiple of the identity. }

 Consider now the  sum of representations that appear in the Clebsch-Gordan decomposition  of 
$(u(d))^{\otimes n}$ (or equivalently of $u^{S_n}(d^n)$ (as from Theorem  \ref{alwaysassumed})). Group them according to the Casimir eigenvalues of the  Casimir operators $C_l$, $l=2,3,...,d$, that is,  write, 
$$
(\mathbb{C}^d)^{\otimes n}=\bigoplus_{j=1}^{n_2} W_{2,j}=\bigoplus_{j=1}^{n_3} W_{3,j}=\cdots=\bigoplus_{j=1}^{n_d} W_{d,j}, 
$$
where, for a fixed  $l=2,...,d$, $W_{l,j}$, $j=1,...,n_l$ is the direct sum of spaces where $C_l$ acts as multiplication by $\lambda_{l,j}$ (the Casimir eigenvalue). with $\lambda_{l,j} \not= \lambda_{l,m}$ if $j \not=m$. Furthermore  
$$
(\mathbb{C}^d)^{\otimes n}=\bigoplus_s V_s, 
$$
where each $V_s$ is the {\it direct sum} of isomorphic irreps of $su(d)$. Fix one $s=\bar s$. $V_{\bar s}$ is a subspace  to only one of the  $W_{2,j}$. It is also a subspace to only one of the $W_{3,j}$'s and so on. Let us assume, without loss of generality and to simplify notations,  that it is always the first one, $W_{l,1}$. Since, the Casimir eigenvalues uniquely identify the irreps, we actually have 
\begin{equation}\label{inclusione5}
V_{\bar s}=\cap_{l=2}^d W_{l,1}. 
\end{equation}
Consider now $C_2$ and, to simplify notations, denote by $\lambda_{j}$, $j=1,...,n_2$ the eigenvalue for which  $W_{2,j}$ is the eigenspace with $\lambda_j \not= \lambda_h$ if $j \not= h$. Denoting by ${\bf 1}_{V}$ the identity on a space $V$, we have for $k=0,...,n_2-1$ 
\begin{equation}\label{C2}
C_2^k=\oplus_{j=1}^{n_2} \lambda_j^k {\bf 1}_{W_{2,j}}. 
\end{equation}
Now choose $(a_1,...,a_{n_2})$ so that 
\begin{equation}\label{tbu}
\sum_{k=1}^{n_2 } a_k \lambda_j^{k-1}=\delta_{1,j}, \qquad j=1,2,...,n_2. 
\end{equation}
This is possible because the determinant of the matrix of coefficients of this system, i.e., 
$$
\det \begin{pmatrix}  1 & \lambda_1 & \lambda_1^2 & \cdots & \lambda_1^{n_2 -1} \cr 
1 & \lambda_2& \lambda_2^2 & \cdots & \lambda_2^{n_2 -1}
\cr 
\cdot & \cdot & \cdot & \cdots & \cdot \cr
\cdot & \cdot & \cdot & \cdots & \cdot \cr
\cdot & \cdot & \cdot & \cdots & \cdot \cr
1 & \lambda_{n_2-1}& \lambda_{n_2-1}^2 & \cdots & \lambda_{n_2-1}^{n_2 -1}
\end{pmatrix}
$$
is a {\it Vandermonde determinant} different from zero since the $\lambda_j$'s are all different from each other. With this choice, we have, using (\ref{tbu}) in (\ref{C2}) 
$$
\sum_{k=1}^{n_2} a_k C_{2}^{k-1}=\oplus_{j=1}^{n_2} \left( \sum_{k=1}^{n_2} a_k \lambda_j^{k-1} \right) {\bf 1}_{W_{2,j}}= \oplus_{j=1}^{n_2} \delta_{1,j}  {\bf 1}_{W_{2,j}}. 
$$
Therefore the (commutative) algebra spanned by the Casimir operators contains the operator which is the identity on $W_{2,1}$ and zero on each of the other $W_{2,j}$'s, $j=2,...,n_2$.

Repeating the same argument for all of the Casimir operators $C_3,...,C_d$, we find that the (commutative) algebra generated by these operators on 
$(\mathbb{C}^d)^{\otimes n}$ contains the operator which is the identity on the intersection of all the $W_{l,1}$, for $l=2,3,...,d$ and zero everywhere else. From (\ref{inclusione5}), we find that such a commutative algebra  contains the operator which is the identity on $V_{\bar s}$ and zero everywhere else. Since $\bar s$ is general the commutative algebra generated by the Casimir operators contains a basis of the center $iu^{S_n}(d^n)$. Since such an algebra is also included in the center, we conclude with the following. 
\begin{proposition}\label{center}
Let ${\cal C}$ be the commutative (Lie) algebra generated by the Casimir operators $C_2,...,C_d$ on $(\mathbb{C}^d)^{\otimes n}$  and ${\cal Z}$ be the center of $u^{S_n}(d^n)$. Then 
$$
i{\cal C}={\cal Z}. 
$$
\end{proposition}

The {\it dimension}  of the center is equal to the number of non-isomorphic irreducible representations of $su(d)$ that appear in the CG decomposition of $(\mathbb{C}^d)^{\otimes n}$.  According to Proposition \ref{clarif} such irreducible representations are parametrized by the $d$-tuples  $(m_1,...,m_d)$ with $m_1 \geq m_2 \geq \cdots \geq m_d$ and $m_1+m_2+\cdots m_d=n$ identifying $d$-tuples $(m_1,...,m_d)$ and $\left( m_1+\lambda ,m_2+\lambda, ...,m_d+\lambda \right)$ for an integer $\lambda$. Such a number of representations can be calculated recursively according to the following proposition. 
\begin{proposition}\label{dimensioncenter}
Let $f(n,d)$ denote the  dimension of the center of $u^{S_n}(d^n)$. Then we have $f(n,1)=1$ for every $n$ and 
\begin{equation}\label{ricorsione}
f(n,d)=\sum_{j=0}^{\lfloor \frac{n}{d}\rfloor} f(n-jd,d-1). 
\end{equation} 
\end{proposition}

\begin{proof}
If $d=1$, independently of what $n$ is, it is clear that there is only a single sequence $(n)$. Thus $f(n,1)=1$. Now for a general $d$, the index $j$ in the sum in (\ref{ricorsione}) represents the value of $m_d$ in $\left( m_1,m_2,...,m_d \right)$. Such a value can go from $0$ to $\lfloor \frac{n}{d}\rfloor$ since a higher value will make it impossible to arrange the remaining $\{ m_1,...,m_{d-1}\}$ all greater or equal to $m_d=j$ and with sum equal to $n-j$.\footnote{More formally, since $m_1$,...,$m_{d-1}$ are all $\geq m_d$, we have 
$n=m_1+m_2+\cdots +m_d\geq d m_d$, so that $m_d \leq \frac{n}{d}$.} Now for a value $m_d=j$, such a value can be subtracted from all entries $(m_1,...,m_d)$ to make the last entry equal to zero again. Once $j=m_d$ is fixed and normalized to zero, the number of possible  $d$-tuples is equal to the number of possible $(d-1)$-tuples  $(m_1,...m_{d-1})$ with $m_1 \geq m_2 \geq \cdots \geq m_{d-1}$ with $m_1+m_2+\cdots + m_{d-1}=n-jd$, that is, exactly $f(n-jd,d-1)$ that appears in the sum in (\ref{ricorsione}). 
\end{proof}

In conclusion and summarizing, we can state that $u^{S_n}(d^n)$ in (\ref{directsum2}) has an orthogonal decomposition 
\begin{equation}\label{orthodec}
u^{S_n}(d^n)=i{\cal C} \oplus su_{cless}^{S_n}(d^n), \qquad su_{cless}^{S_n}(d^n):=\bigoplus_{j=1}^{f(n,d)} su(n_j), 
\end{equation}
where ${\cal C}$ is the center which coincides with the algebra generated by the Casimir operators and $n_j$ denotes the dimension of the $j$-th irreducible module in the Clebsch-Gordan decomposition 
which is given by formula (\ref{dimensione}), while the number of different (non-isomorphic) representations (which coincides with the dimension of ${\cal C}$)  is given by $f(n,d)$ in Proposition \ref{dimensioncenter} and their nature is described in Proposition \ref{clarif}. Decomposition (\ref{orthodec}) is a Levi's type (see, e.g., \cite{Knapp}) of decomposition for the Lie algebra $u^{S_n}(d^n)$ as a direct sum of its center and a semisimple Lie algebra.

\section{The dynamical Lie algebra generated by symmetric Hamiltonians}\label{dynLA}

Our interest in $u^{S_n}(d^n)$ is due to the fact that the dynamical Lie 
algebra ${\cal L}$ of any quantum control system described by Hamiltonians $\{H_1, H_2,...,H_m\}$ which are  invariant under the action of the  permutation group, that 
is,  any ${\cal L}=\{ H_1,...,H_r\}_{Lie}$,\footnote{Here and in the following, we denote by $\{A_1,..,A_r\}_{Lie}$ the Lie algebra generated by a set $\{A_1,...,A_r\}$.}   is a Lie subalgebra of  
$u^{S_n}(d^n)$. It seems also appropriate to recall at this point that a  (quantum) dynamical Lie algebra on ${\cal H}=(\mathbb{C}^d)^{\otimes n}$, which is necessarily a subalgebra of $u(d^n)$,  is always {\it reductive}, that is, it is the direct (commuting) sum ${\cal L}={\cal L}_C \oplus {\cal L}_S$, of its center ${\cal L}_C$ and a semisimple part ${\cal L}_S$, generalizing (\ref{orthodec}).\footnote{This is a known fact in the theory of non-controllable quantum systems  (see, e.g., \cite{mikobook}).}   In characterizing the dynamical Lie algebra ${\cal L}$ we shall use the facts in the following two  theorems for Lie algebras, which we state in general terms.

\begin{theorem}\label{genfacLem}
Consider a general (real or complex) (ambient) Lie algebra ${\cal G}$ which has a Levi direct sum decomposition 
\begin{equation}\label{deco57}
{\cal G}={\cal C} \oplus \hat {\cal S}, 
\end{equation}
where ${\cal C}$ is the center and $\hat {\cal S}$ a (semisimple) Lie algebra.
Consider a set of linearly independent elements in ${\cal G}$ written as $\{ H_k:= C_k+S_k, \, k=1,...,r\}$ with $C_k \in {\cal C}$ and $S_k \in \hat {\cal S}$ for every $k=1,...,r$. Assume $\{S_1,...,S_r\}_{Lie}$ is semisimple. Then 
\begin{equation}\label{POIL}
\{ C_1+S_1, ..., C_r+S_r\, \}_{Lie}=\texttt{span}\{C_1,...,C_r\} \oplus \{ S_1,...,S_r\}_{Lie}. 
\end{equation}
\end{theorem}

\begin{proof}
We have 
\begin{equation}\label{yyt2}
\{ C_1+S_1,...,C_r+S_r\}_{Lie}= \texttt{span}\left\{ C_1+S_1,C_2+S_2,...,C_r+S_r, {\cal N}_1 \right\}, 
\end{equation}
where ${\cal N}_1$ is a basis in the space of (nested) commutators of $\{ C_1+S_1,...,C_r+S_r\}$. Analogously, we have 
\begin{equation}\label{yyt}
\{ S_1,...,S_r\}_{Lie}= \texttt{span}\left\{ S_1,S_2,...,S_r, {\cal N}_2 \right\}, 
\end{equation}
where ${\cal N}_2$ is a basis of (nested) commutators of  $\{ S_1,...,S_r\}$. Since $\{C_1,...,C_r\}$ belong to the center of ${\cal G}$, we can choose ${\cal N}_1={\cal N}_2:={\cal N}$. Furthermore, since $\{ S_1,...,S_r\}_{Lie}$ is assumed to be semisimple\footnote{Recall (see, e.g., \cite{Knapp} that one of the equivalent definitions of semisimple Lie algebra  ${\cal L}$, is that 
$\left[ {\cal L}, {\cal L}  \right]= {\cal L}$, that is, there exists a basis of commutators.} we have that for each $j=1,...,r$, $S_j \in \texttt{span}\{{\cal N} \}$ and therefore, from (\ref{yyt}), 
\begin{equation}\label{uiop}
\{ S_1,...,S_r\}_{Lie}=\texttt{span} \{ {\cal N} \}. 
\end{equation}

Semisimplicity of $\{S_1,...,S_r\}_{Lie}$ also implies that in (\ref{yyt2}) we have\footnote{Since $\{S_1,...,S_r\} \subseteq   {\cal N} $.}
$$
\{ C_1+S_1,...,C_r+S_r\}_{Lie}= \texttt{span}\{ C_1,C_2,...,C_r, {\cal N} \}=
\texttt{span} \{ C_1,C_2,...,C_r \} \oplus \texttt{span} \{ {\cal N} \}=
$$
$$
\texttt{span} \{ C_1,C_2,...,C_r \} \oplus \{ S_1,...,S_r \}_{Lie}, 
$$
using (\ref{uiop}), as desired. 
\end{proof}

\begin{theorem}\label{another}
Reconsider the situation of Theorem \ref{genfacLem}. The following statements are equivalent. 
\begin{enumerate}
\item 
$\hat {\cal S} \subseteq \{ C_1+S_1,...,C_r+S_r\}_{Lie}$
\item  $\{ S_1,...,S_r\}_{Lie}=\hat {\cal S}$
\item 
$\{ C_1+S_1,...,C_r+S_r \}_{Lie}=\texttt{span} \{C_1,...,C_r\} \oplus \hat {\cal S}$.
\end{enumerate}

\end{theorem}

\begin{proof}
$2 \Rightarrow 3$ follows as a special case of  the statement of Theorem \ref{genfacLem} since $\hat {\cal S}$ is semisimple,  while $3 \Rightarrow 1$ is obvious. To prove that $1 \Rightarrow 2$ assume that there is $S$ in $\hat {\cal S}$ orthogonal\footnote{We are assuming here that we are dealing with skew-Hermitian matrices  with the Frobenius inner product $\langle A, B \rangle:=Tr(AB^\dagger)$.} to $\{S_1,...,S_r\}_{Lie}$ and t	herefore orthogonal to  $\{S_1,...,S_r\}$. This element would be orthogonal to $\{C_1+S_1,...,C_r+S_r\}$ also and therefore to whole  $\{C_1+S_1,...,C_r+S_r\}_{Lie}$ which contradicts $1$. 
\end{proof}

The use of Theorems \ref{genfacLem} and \ref{another} in our case is as follows: For us ${\cal G}$ in 
(\ref{deco57}) is the Lie algebra $u^{S_n}(d^n)$ in (\ref{directsum2}), (\ref{orthodec}), 
and $\hat {\cal S}$ is $su_{cless}^{S_n}(d^n)$  in (\ref{orthodec})  which is the  direct sum of the simple $su(n_j)$ for different $n_j$ and therefore semisimple. The center ${\cal C}$ is $i {\cal C}$ in (\ref{orthodec}), that is, the Casimir subalgebra. When checking the 
controllability for a set of generators $\{H_1,...,H_r\}$, we can split the generators 
as in (\ref{POIL}), $H_j=C_j+S_j$, $j=1,...,r$, and check, as in Theorem \ref{genfacLem}, that 
$\{S_1,...,S_r\}$ generates $su_{cless}^{S_n}(d^n)$ or equivalently, without doing the splitting, verify condition 1 of Theorem \ref{another}. We shall refer to this condition as {\bf subspace controllability}. If subspace controllability is verified, then the dynamical Lie algebra is the direct sum of  $su_{cless}^{S_n}$ in (\ref{orthodec}) and an Abelian subalgebra spanned by the components of the generators onto the Casimir Lie algebra (the center of $u^{S_n}(d^n)$). Therefore the two main results of \cite{Marco},\footnote{These results concern special types of  permutation invariant Hamiltonians (cf. next Section \ref{caseofnqubits}),},  that is, the subspace controllability and the, so called, `{\it failure of universality}', that is, the impossibility to control the relative phases between evolutions   on two different invariant subspaces are not really independent results. They are connected to each other according to  Theorems \ref{genfacLem} and \ref{another}. In the following section we revisit these results within our framework. Afterwards, in  section \ref{threequtrits} we shall treat a new case of three qutrits.

In closing this section, we elaborate a bit more on the concept of subspace controllability. In general, for a system in block diagonal form, say with $m$ blocks, it may be sufficient that once we fix a block, of dimension say $n_1$, the dynamical Lie algebra ${\cal L}$ is such that every element of $su(n_1)$ can be obtained on that block. We may call such property {\it weak} subspace controllability. Weak subspace controllability is not the same as subspace controllability as for example a Lie algebra ${\cal L}$ of matrices of the form 
$\begin{pmatrix} A & 0 \cr 0 & A \end{pmatrix}$ with $A$ arbitrary in $su(n_1)$ is weak subspace controllable but not subspace controllable, since the two diagonal blocks are related,\footnote{Another example, in the case of the CG decomposition,  blocks that correspond to isomorphic irreducible representations also are related.} they are equal. However, the following Lemma shows that in the case of blocks of different dimensions, the two properties coincide. 

\begin{lemma}\label{lemmaneeded}
Consider  a Lie algebra ${\cal R}$ of block diagonal matrices 
in $su(n_1+n_2+\cdots+n_r)$ with blocks of dimension $n_1$, $n_2$,...,$n_r$.  Assume that weak subspace controllability is verified and  
$$
n_j \not= n_l \quad \texttt{ if} \quad   j \not=l. 
$$
Then 
$$
{\cal R}=\bigoplus_{j=1}^r su(n_j). 
$$

\end{lemma}

\begin{proof} According to the assumptions of the Lemma, matrices in ${\cal R}$ have the form $\texttt{diag} \left( A_1, A_2,...,A_r \right)$ with $A_j \in su(n_j)$ for $j=1,...,r$,  and,  if we fix an arbitrary $j\in \{1,2,...,r\}$ and an arbitrary  $Z \in su(n_j)$,  we can find a matrix  
$\texttt{diag}\left( A_1,A_2,...,A_r\right) \in {\cal R}$ with $A_j=Z$. Without loss of generality, let us assume that the blocks of the matrices in ${\cal R}$ are ordered according to the dimensions  in strictly decreasing order, from top to bottom,  that is, 
$$
n_1> n_2>\cdots > n_r. 
$$  
We shall prove that ${\cal R}$ contains matrices of the form 
$$
\texttt{diag} \left(L_1,0,...,0 \right), \qquad \texttt{diag} \left( 0,L_2,0,...,0 \right),..., 
\texttt{diag} \left( 0,0,...,0,L_r \right), 
$$
for arbitrary $L_j \in su(n_j)$, which proves the claim.

Consider the first block (the one with dimension $n_1$ which is the highest) and consider $n_1-1$ linearly independent matrices forming a basis of the  Cartan subalgebra (CSA) (maximal Abelian subalgebra) of $su(n_1)$, $\bar A_1,...,\bar A_{n_1-1}$.\footnote{Recall that associated with $su(N)$ is the {\it rank},  $N-1$, which is the largest number of mutually commuting linearly independent matrices in $su(N)$, which forms a basis of a CSA (see, e.g., \cite{Knapp} for details).} Therefore ${\cal R}$ contains matrices 
$$
K_1:=\texttt{diag}(\bar A_1,B_1,*,...,*), \quad K_2:=\texttt{diag}(\bar A_2,B_2,*,...,*),...,\quad 
K_{n_1-1}:=\texttt{diag}(\bar A_{n_1-1},B_{n_1-1},*,...,*). 
$$  
There are two possibilities 

\begin{itemize}
\item {\bf Case 1} There exist two indexes $j$ and $k$ such that $B_j$ and $B_k$ do {\bf not} commute. 

\item {\bf Case 2} The matrices $\{B_1,...,B_{n_1-1}\}$ in $su(n_2)$ {\bf all} commute. 

\end{itemize}

In {\bf Case 1}, we have 
$$
[K_j,K_k]=\texttt{diag}(0, [B_j,B_k],*,...,*) . 
$$
Therefore ${\cal R}$ contains the vector space sum as $l$ goes from $0$ to $\infty$ of 
$\texttt{diag} \left( 0, ad_{su(n_2)}^l[B_j,B_k], *, ...,* \right)$. The sum as $l$ goes from $0$ to $\infty$ of $ad_{su(n_2)}^l[B_j,B_k]$ is the smallest ideal in $su(n_2)$ containing $[B_j, B_k]$ and since $[B_j, B_k]\not=0$  and $su(n_2)$  is simple  it is equal to 
$su(n_2)$. Therefore ${\cal R}$ contains matrices of the form 
$
\texttt{diag} \left(0, B,*,...,* \right), 
$
with arbitrary $B \in su(n_2)$, and, since it contains also matrices $\texttt{diag} \left(\bar A, *,*,...,* \right)$ with arbitrary $\bar A \in su(n_1)$,  it contains matrices of the form $\texttt{diag}(\bar A,0,*,...,*)$ with arbitrary $\bar A \in su(n_1)$.

In {\bf Case 2}  with all matrices $\{ B_1,...,B_{n_1-1} \} \in su(n_2)$ commuting, since the rank of $su(n_2)$ is $n_2-1 < n_1-1=\texttt{rank}(su(n_1))$, $\{ B_1,...,B_{n_1-1} \} $ are linearly dependent and therefore there exists a nontrivial linear combination $\sum_{j=1}^{n_1-1} b_j B_j=0$. Therefore ${\cal R}$ contains all matrices
$$
\sum_{j=1}^{n_1-1} b_j K_j=\texttt{diag} \left(  \sum_{j=1}^{n_1-1} b_j \bar A_j, \sum_{j=1}^{n_1-1} b_j B_j, *,...,*  \right)= \texttt{diag} \left(  \sum_{j=1}^{n_1-1}  b_j \bar A_j, 0 , *,...,*  \right)=\texttt{diag} \left( \tilde A, 0 , *,...,*  \right), 
$$ 
 where we have defined $\tilde A:=\sum_{j=1}^{n_1-1} b_j \bar A_j$. Since the $\bar A_j$'s are linearly independent, and the $b_j$ are not all zeros, we have $\tilde A \not= 0$. Since $\cal R$ contains  
 diagonal blocks where the first block can be chosen 
 arbitrarily in $su(n_1)$, we have that ${\cal R}$ contains all matrices $\texttt{diag} \left( \bar A , 0, *,...,* \right)$ with arbitrary $\bar A$ in the sum as $l$ goes from $0$ to $\infty$ of the $ad_{su(n_1)}^l  \tilde A$ which is $su(n_1)$. Therefore we reach the same conclusion as in Case 1.

\vspace{0.25cm} 

Repeating the same argument with block 1 and block 3, with block 3 taking the role of block $2$, we obtain that ${\cal R}$ contains all matrices of the form $\texttt{diag} \left(\bar A, 0, 0, *, ...,* \right)$ with arbitrary $\bar A \in su(n_1)$. Continuing this way for the following blocks, $4,5,...,r$,  if any, we reach the conclusion that ${\cal R}$ contains matrices of the form  
$\texttt{diag}\left( L_1,0,...,0 \right)$ with arbitrary $L_1 \in su(n_1)$. Since ${\cal R}$ contains matrices $\texttt{diag} \left(A_1,A_2,*,...,*\right)$ with arbitrary $A_2 \in su(n_2)$, it follows that it also contains matrices of the form $\texttt{diag} \left( 0, \bar A_2,*,...,*\right)$ with arbitrary $\bar A_2 \in su(n_2)$. We can therefore now repeat the same argument as above with block 2 replacing block 1 and going down the other blocks. We conclude that ${\cal R}$ contains matrices of the form $\texttt{diag}\left(0, L_2, 0,0,...,0  \right)$ with arbitrary $L_2 \in su(n_2)$. We then continue for blocks $3$, $4$, and so on (if any) to complete the proof of the Lemma.

\end{proof}

\begin{remark}\label{practical5} 
The practical use of Lemma \ref{lemmaneeded} is that, in trying to prove subspace controllability for a certain system, we can focus on each invariant subspace separately. Another use is that, if we want to prove controllability on a certain vector space $V$ and we are able to identify two invariant subspaces for a certain {\em subset}  of generators, say $V=V_1 \oplus V_2$ we can prove controllability on $V_1$ and then $V_2$, separately,  so that we have arbitrary blocks and then use the remaining generators to `connect' the two blocks, thus having the full $su(N)$ for arbitrary $N$. We shall use this idea, along with Facts 1 and 2 of subsection \ref{duetre} in section \ref{threequtrits} to prove subspace controllability for the case of three qutrits. 

\end{remark}

\section{The case of $su(2)$; $n$ qubits}\label{caseofnqubits}

In the case of $su(2)$, the Clebsch-Gordan  decomposition  coincides with the standard one obtained in quantum mechanics in the theory of decomposition of angular momentum (see, e.g., \cite{Sakurai}, \cite{Tung}). The non isomorphic representations that appear in the decomposition are indexed by a single quantum (spin) number $p$ and are the ones that appear in (\ref{ponly}). 
There are $n_2:=\lfloor\frac{n}{2}\rfloor+1$ of them and this is the dimension of the center   of $u^{S_n}(d^n)$, which, according to Theorem \ref{center}, is given by $i{\cal C}$, where a basis of 
${\cal C}$ is 
$$
C_2^0={\bf 1}, C_2, C_2^2,...,C_{2}^{n_2-1}=C_2^{\lfloor \frac{n}{2}\rfloor}, 
$$
with $C_2$ the quadratic Casimir operator  (\ref{Casimirsu2}), (\ref{Casimirsu2plus}). As an orthonormal basis of $iu(2)$, when calculating (\ref{Casimirsu2plus}), it is customary to take the normalized identity and normalized Pauli matrix. To ease notations,  we take such matrices without the normalization factor so as they form an {\it orthogonal}  rather than {\it orthonormal} basis of $iu(2)$; the results will only differ by an unimportant common factor.  We recall the definition of the Pauli matrices 
\begin{equation}\label{Paulimat}
\sigma_x=  \begin{pmatrix} 0 & 1 \cr 1 & 0 \end{pmatrix}, \qquad 
\sigma_y=  \begin{pmatrix} 0 & i \cr -i & 0 \end{pmatrix}, \qquad 
\sigma_z=  \begin{pmatrix} 1 & 0 \cr 0 & -1 \end{pmatrix}. 
\end{equation}
Let us also use the notation of Theorem \ref{alwaysassumed} by denoting by $F_{(j_0,j_1,j_2,j_3)}$ the symmetric sum of all the tensor products with $j_0$ identities, $j_1$, $\sigma_x$'s, $j_2$, $\sigma_y$'s, $j_3$, $\sigma_z$'s. Then formula (\ref{Casimirsu2plus}) gives 
\begin{equation}\label{newdefC2}
C_2:=F_{(n-1,1,0,0)}^2+F_{(n-1,0,1,0)}^2+F_{(n-1, 0 , 0 ,1)}^2.
\end{equation}
In particular 
for $F_{(n-1,1,0,0)}^2$ we have\footnote{We use here  the standard convention of denoting by $\sigma_{x,y,z}^j$ the matrix which is the tensor product of $n$ identities except in position $j$ which is occupied by $\sigma_{x,y,z}$.} 
\begin{equation}\label{56tb}
F_{(n-1,1,0,0)}^2=\left( \sum_{j=1}^n \sigma_x^j\right)^2
=\sum_{j=1}^n
(\sigma_x^j)^2+2 \sum_{1\leq j <k \leq n} \sigma_x^j \sigma_x^k=
n {\bf 1} + 2 \sum_{1\leq j <k \leq n} \sigma_x^j \sigma_x^k=n {\bf 1} + 2 F_{(n-2,2,0,0)}, 
\end{equation}
and analogous formulas for $F_{(n-1,0,1,0)}$, $F_{(n-1,0,0,1)}$, and summing the terms 
(for $x$, $y$, and $z$), according to (\ref{newdefC2}),  we obtain 
\begin{equation}\label{stilCasimir}
C_2=3n {\bf 1}+ 2 \left( F_{(n-2,2,0,0)}+ F_{(n-2,0,2,0)} +F_{(n-2,0,0,2)} \right)=3n{\bf 1} + 2 A, 
\end{equation}
defining $A:=F_{(n-2,2,0,0)}+ F_{(n-2,0,2,0)} +F_{(n-2,0,0,2)}$. Therefore, a new `direction' with respect to $C_2^{0}={\bf 1}$ is in $A$. The matrix $2A$ is the matrix $C_1$ given in formula (B.2) of \cite{Marco}. In general,\footnote{We can use the binomial formula since the matrices ${\bf 1}$ and $A$ commute.} for $k=0,1,...,\lfloor \frac{n}{2} \rfloor$ 
\begin{equation}\label{bino1}
C_2^k=\left(3n{\bf 1}+ 2A \right)^k=\sum_{j=0}^k {{k}\choose{j}} \left( 3n \right)^{k-j} 2^j A^j, 
\end{equation}
and directions not in $\texttt{span} \{ C_2^0, C_2,...,C_2^{k-1} \}$ are in $A^k$, that is, $A^k$ is the sum of a term in 
$\texttt{span} \{ C_2^0, C_2,...,C_2^{k-1} \}$ and an extra term which is    proportional to the  element $C_k$ in formula (B2), (B3)  of reference \cite{Marco}. This way one obtains the elements of the basis of the center given in reference 
\cite{Marco} which we repeat here 
\begin{equation}\label{FormulaB2}
\sum_{a+b+c=k}\frac{(2a)! (2b)! (2c)!}{a! b! c!} F_{(n-2k,2a,2b,2c)}, \qquad 0 \leq k \leq \lfloor\frac{n}{2}\rfloor
\end{equation}

The dynamical Lie algebra considered in \cite{Marco} (see also \cite{confraJMP}) is generated by the {\it local} Hamiltonians $F_{(n-1,1,0,0)}$, $F_{(n-1,0,1,0)}$, and $F_{(n-1,0,0,1)}$ and, for fixed $k$,  by $j$-body symmetric Hamiltonians for $j=2,...,k$,  $F_{(n-j,0,0,j)}.$ Since, according to (\ref{FormulaB2}) the center only contains symmetric $j$-body Hamiltonians with $j$ even, the only Hamiltonians which have nonzero component on the center of $u^{S_n}(d^n)$ are the  
$
F_{n-j,0,0,j}$'s, with $j$ even. Therefore,  assuming that the component of the generators in 
$su_{cless}^{S_n}(2^n)$ generate all 
of $su_{cless}^{S_n}(2^n)$, from Theorem  \ref{genfacLem}, 
we have that the dynamical Lie algebra is $su_{cless}^{S_n}(2^n) \oplus Q$, where $Q$ is the (commutative) span of the components of the generators along the center  which in the case of 
at most $k$-body Hamiltonians has dimension 
$\lfloor \frac{k}{2} \rfloor$. It is indeed true that the generators generate all of $su_{cless}^{S_n}(2^n)$. This fact is shown in \cite{Marco}.\footnote{The proof in \cite{Marco} uses   a similar   technique to what  was put forward in \cite{confraJMP}.}
We record below  this result  in a different format for future use and within the framework of the present work. 

\begin{theorem}\label{theordiffformat}

Consider a Clebsch-Gordan (reducible) representation of $su(2)$ on $n$ qubits and 
let $\hat S_x,\hat S_y,\hat S_z$ the images (under the homomorphism (\ref{totalOperators}) (\ref{rrt})  defining the 
representation) of $\sigma_x,\sigma_y,\sigma_z$  (\ref{Paulimat}). Denote by $V^j$, 
$j=1,...,\lfloor \frac{n}{2}\rfloor +1$ the irreducible modules of the (nonisomorphic) representations, appearing in the (reducible) representation. Consider the Lie algebra ${\cal L}$ generated by $i \hat S_x,i \hat S_y,i \hat S_z$ and $i \hat S_z^2$. Then 
\begin{equation}\label{AnnikaLeo}
{\cal L}:= \left( \bigoplus_{j=1}^{\lfloor \frac{n}{2} \rfloor +1}  su \left( \texttt{dim} (V^j) \right) \right) 
\oplus \texttt{span} \{ \left(  in {\bf 1}+i2 F_c\right) |_{\oplus_j V^j}\},  
\end{equation}
where $F_c$ is the orthogonal component of $F_{({n-2}, 0, 0, 2)}$ onto the center ${\cal C}$, and $A|_V$ denotes the restriction of the operator $A$ on the subspace $V$.\footnote{Note that this is proportional to (\ref{FormulaB2}) with $k=1$ since this is the only part with $2$-body Hamiltonians in the center. Notice also that $in {\bf 1}+i2 F_c$ acts as a scalar operator on each of the subspaces $V^j$.} 
\end{theorem}
\begin{remark}\label{AnnikaLeoRem} It is important to remark, and it will be used in the next section,  that the above result holds even if we do not consider the {\it full} Clebsch-Gordan  decomposition of $(\mathbb{C}^2)^{\otimes n}$ but only some of the invariant irreducible modules. These are  (possibly) reducible representation where all the irreducible modules $V^j$ also appear in a  Clebsch-Gordan decomposition of $(\mathbb{C}^2)^{\otimes n}$, for some $n$. Formula (\ref{AnnikaLeo})  in these cases will give a direct sum of  $su \left( \texttt{dim} (V^j) \right) $ Lie algebras, with some positive integer $m$ replacing $\lfloor \frac{n}{2} \rfloor+1$. This is  also the reason why we chose the notation in (\ref{AnnikaLeo}) rather than simply write $\texttt{span} \{ \left(  in {\bf 1}+i2 F_c\right) \}$. 

\end{remark}

\begin{proof} The result of \cite{confraJMP}, \cite{Marco} says that the Lie algebra generated by $i \hat S_{x,y,z}$ and $iF_{(n-2,0,0,2)}$ on $(\mathbb{C}^2)^{\otimes n}$ is the direct sum 
of all the $su\left( \dim(V) \right)$'s for every $V$  irreducible Clebsch-Gordan module,  $\oplus$ the span of the operator in formula (\ref{FormulaB2})  for  $k=1$, which is $A$ in formula (\ref{stilCasimir}). Write $F_{(n-2,0,0,2)}$ according to the orthogonal decomposition 
$$
F_{(n-2,0,0,2)}:=F_s+F_c, 
$$
where $F_s$ is the component onto $\bigoplus_j  i su \left( \dim ( V^j) \right) $ 
where the sum is taken over {\it all} the (nonisomorphic) Clebsch-Gordan subspaces  $V^j$ and $F_c$ is the orthogonal complement in $u^{S_n}\left( d^n \right)$, that is, the component 
along the center. The result of \cite{confraJMP}, \cite{Marco} says that 
\begin{equation}\label{risultato}
\left\{ i F_{(n-2,0,0,2)}, i \hat S_x, i \hat S_y,i \hat S_z \right\}_{Lie}=\left\{ i F_s+i F_c, i \hat S_x, i \hat S_y, i \hat S_z \right\}_{Lie}= 
\texttt{span} \{i  F_c\} \oplus_{j} su \left( \dim  V^j\right). 
\end{equation}
We remark that this implies that\footnote{If this was not true,  then there would be a matrix 
$A \in \oplus_{j} su \left( \dim ( V^j  
) \right)$ orthogonal to $\{i  F_s,i \hat  S_x,i \hat S_y, i \hat S_z\}_{Lie}$ and 
therefore to $\{i F_c+ i F_s, i \hat S_x, i \hat S_y, i \hat S_z\}_{Lie}$ which contradicts the result (\ref{risultato}).}  
\begin{equation}\label{impoeq}
\{ i F_s, i \hat S_x,i \hat S_y,i \hat S_z\}_{Lie}=
\oplus_{j} su \left( \dim V^j\right).
\end{equation} 
Now using (\ref{56tb}) we get 
$$
\{ i \hat S_z^2, i \hat S_x, i \hat S_y, i \hat S_z\}_{Lie}=\{in{\bf 1} +2iF_s+2iF_c, i \hat S_x,i \hat S_y, i \hat S_z \}_{Lie}, 
$$
which, using Theorem \ref{genfacLem} with (\ref{impoeq}), and reducing (possibly) only to certain irreducible modules gives  (\ref{AnnikaLeo}). 
\end{proof}

\section{Subspace controllability for a symmetric network of three qutrits}\label{threequtrits}

We now consider, and prove, subspace controllability 
for a symmetric network of $n=3$ three 
dimensional quantum systems ({\it qutrits}). This is a new case. Some of the things we shall say are actually valid for general $n$ and we believe could be used for systems of general dimensions $d$ as well.  Generalizing the situation described for $su(2)$, the dynamical Lie algebra we shall study is generated by arbitrary `local' symmetric Hamiltonians of the form 
$$
H=L\otimes {\bf 1} \otimes \cdots \otimes {\bf 1} + {\bf 1} \otimes L \otimes {\bf 1} \cdots {\bf 1} + \cdots + {\bf 1} \otimes \cdots \otimes {\bf 1} \otimes L, 
$$
for general $L$ in $isu(3)$ and by a 2-body Hamiltonian which is a symmetric sum of all tensor products containing the identity  in every location except in two locations which are occupied by a given $\tilde E \in isu(3)$. For definiteness, we shall 
take  for $\tilde E=S_z^1:=E_3$ in (\ref{GellMann}), 
(\ref{Paulimatplus2}) below.\footnote{The case of $\tilde E$ 
any matrix in $i su(3)$ with a zero eigenvalue  can be reduced to this case.}

\subsection{Generalities about $su(3)$}
For $su(3)$, the role of the Pauli matrices as a basis of $isu(3)$ is taken by the {\it Gell-Mann} matrices  which we report here 
\begin{equation}\label{GellMann}
E_1:=\begin{pmatrix} 0 & 1 & 0 \cr 1 & 0 & 0 \cr 0 & 0 & 0  \end{pmatrix}, \, 
E_2:=\begin{pmatrix} 0 & -i & 0 \cr i & 0 & 0 \cr 0 & 0 & 0  \end{pmatrix}, \, 
E_3:=\begin{pmatrix} 1 & 0 & 0 \cr 0 & -1 & 0 \cr 0 & 0 & 0  \end{pmatrix}, \, 
E_4:=\begin{pmatrix} 0 & 0 & 1 \cr 0 & 0 & 0 \cr 1 & 0 & 0  \end{pmatrix},
\end{equation}
$$
E_5:=\begin{pmatrix} 0 & 0 & -i \cr 0 & 0 & 0 \cr i & 0 & 0  \end{pmatrix}, \, 
E_6:=\begin{pmatrix} 0 & 0 & 0 \cr 0 & 0 & 1 \cr 0 & 1 & 0  \end{pmatrix} \, 
E_7:=\begin{pmatrix} 0 & 0 & 0 \cr 0 & 0 & -i \cr 0 & i & 0  \end{pmatrix} \, 
E_8:=\frac{1}{\sqrt{3}} \begin{pmatrix} 1 & 0 & 0 \cr 0 & 1 & 0 \cr 0 & 0 & -2  \end{pmatrix}.  
$$
They satisfy the commutation and anti-commutation relations 
\begin{equation}\label{commuanti}
[E_j, E_k]=\sum_{l} i f_{j,k}^l E_l, \qquad \{ E_j, E_k \} =\frac{4}{3} {\bf 1}+ \sum_{l} d_{j,k}^l E_l,
\end{equation}
where $f_{j,k}^l$ are the {\it structure constants} which are {\it antisymmetric} with respect to any change of the indexes and $d_{j,k}^l $ are  {\it symmetric} with respect to any change of the indexes. The $f_{j,k}^l$ and $d_{j,k}^l$ constants are easily computed \cite{QMS} and we report them along with some of their properties in Appendix \ref{GM}.

 \vspace{0.25cm}
 
The {\it raising} and {\it lowering} operators  $S_+^l$ and $S_-^l$ $l=1,2$,   are defined specializing (\ref{Szpmstandard}) (\ref{totalOperators}). Furthermore, (cf. (\ref{rrt})) in analogy with the Pauli matrices (\ref{Paulimat}), we define\footnote{These definitions naturally extend to $su(d)$ (see, e.g., \cite{Alex}.} 
\begin{equation}\label{Paulimatplus}
S_x^1:=S_+^1+S_-^1=E_1, \, S_x^2:=S_+^2+S_-^2=E_6, \, 
S_y^1:=i\left(S_+^1-S_-^1 \right)=-E_2, \,  S_y^2:= i\left(S_+^1-S_-^1 \right)=-E_7, 
\end{equation}
 and
\begin{equation}\label{Paulimatplus2}
 S_z^1:=E_3, \qquad S_z^2=\begin{pmatrix} 0 & 0 & 0 \cr 0 & 1 & 0 \cr 0 & 0 & -1  \end{pmatrix}. 
\end{equation}

\subsection{The center of $u^{S_n}(3^n)$ and calculation of the quadratic Casimir operator}
 
 We know from Proposition \ref{center} that the center of $u^{S_n}(3^n)$ is the Abelian Lie algebra generated by the Casimir operators. In the case of $su(3)$, as compared with the case of $su(2)$ however, there are  two independent Casimir operators $C_2$ and $C_3$, needed to distinguish between two different irreps of $su(3)$. In our case, not all the possible irreducible representations appear in the Clebsch-Gordan decomposition of $(\mathbb{C}^3)^{\otimes n}$.  However, in general, one still needs both the  quadratic and cubic Casimir eigenvalues to distinguish between representations.  This follows from the analysis in Appendix \ref{quadraticCasi} (cf., in particular Example \ref{controesempio1}). In some cases, like the case of $n=3$ which will be treated next,  the quadratic Casimir  eigenvalue is sufficient to distinguish between different representations. Therefore,  in these cases (cf. the proof of Proposition \ref{center}), the center of $u^{S_n}(3^n)$ is given by the algebra generated by the quadratic Casimir operator, that is, it is given,  by the span of 
  $$
 C_2^0={\bf 1}, C_2,...,C_2^{n_2-1}. 
 $$
Here  $n_2$  is the dimension of the center, given by the number of non isomorphic irreducible representations in $(\mathbb{C}^3)^{\otimes n}$,  which is, using (\ref{admissrep}):\footnote{Obtaining a closed form expression for $n_2$ in the case of $su(3)$ is more difficult than in the case of $su(2)$. From formula (\ref{n2su3}), we obtain 
 $$
 n_2=\lfloor \frac{n}{3} \rfloor +1 +  \sum_{m=n,n-3,...,n-3 \lfloor \frac{n}{3}\rfloor, \, m \texttt{ even}} \left( \lfloor \frac{m}{2}  \rfloor \right)+\sum_{m=n,n-3,...,n-3 \lfloor \frac{n}{3}\rfloor, \, m \texttt{ odd}} \left( \lfloor \frac{m}{2}  \rfloor \right)=
  $$
  $$
  \lfloor \frac{n}{3} \rfloor +1 +  \sum_{m=n,n-3,...,n-3 \lfloor \frac{n}{3}\rfloor, \, m \texttt{ even}} \left(  \frac{m}{2}   \right)+\sum_{m=n,n-3,...,n-3 \lfloor \frac{n}{3}\rfloor, \, m \texttt{ odd}} \left(  \frac{m-1}{2}   \right)=
  $$
  $$
   \lfloor \frac{n}{3} \rfloor +1 + \frac{1}{2} \left(  \sum_{m=n,n-3,...,n-3 \lfloor \frac{n}{3}\rfloor} m -NUM(n) \right), 
  $$
  where $NUM(n)$ denotes  the number of odd numbers in the set $\{n,n-3,n-6,...,n-3\lfloor \frac{n}{3} \rfloor \}$. 
  This gives
  $$
  n_2=\lfloor\frac{n}{3}\rfloor +1 +\frac{1}{2} \left( \sum_{j=0}^{\lfloor \frac{n}{3}\rfloor}  (n-3j) -NUM(n)\right)=
  \lfloor\frac{n}{3}\rfloor +1 +\frac{1}{2} \left( n\left( \frac{n}{3} +1\right)-3 \sum_{j=1}^{\lfloor\frac{n}{3}\rfloor} j -NUM(n) \right) =
  $$
  $$
  \lfloor \frac{n}{3} \rfloor+1 +\frac{1}{2} \left( n \left( \lfloor \frac{n}{3} \rfloor+1 \right) -\frac{3}{2} \left( \lfloor \frac{n}{3}\rfloor+1\right)\lfloor\frac{n}{3} \rfloor -NUM(n) \right). 
  $$
This expression can be slightly simplified by writing $n=3 \lfloor \frac{n}{3} \rfloor+l$ for $l=0,1,2$. This gives  
  $$
n_2=  \lfloor \frac{n}{3} \rfloor+1 +\frac{1}{2} \left(    \left( \lfloor\frac{n}{3} \rfloor +1 \right) \left( \frac{3}{2} \lfloor\frac{n}{3} \rfloor +1\right) -NUM(n)\right).
$$
Furthermore, with this notation $NUM(n)=\frac{1}{2} \lfloor \frac{n}{3} \rfloor +\alpha$, where $\alpha=1$ if $\lfloor\frac{n}{3} \rfloor $ is even and $l$ is odd, 
 $\alpha=\frac{1}{2}$ if $\lfloor\frac{n}{3} \rfloor $ is odd,  
 $\alpha=0$ if $\lfloor\frac{n}{3} \rfloor$ is even and $l$ is odd.}
\begin{equation}\label{n2su3}
 n_2=\sum_{m=n,n-3,...,n-3\lfloor \frac{n}{3}\rfloor} \left( \lfloor \frac{m}{2} \rfloor +1  \right). 
 \end{equation}

\vspace{0.25cm} 

Since in some cases (like the $n=3$ case below) the quadratic Casimir eigenvalue is sufficient to generate the whole center of $u^{S_n}(3^n)$ and to be in parallel with the case of $d=2$ treated in the previous section, we proceed to the calculation of the quadratic Casimir operator  for the case $d=3$.


Consider again the matrices $F_{(j_0,j_1,...,j_8)}$, the symmetric sum of all tensor products containing $j_0$ times $E_0:={\bf 1}$, $j_1$ times $E_1$ in (\ref{GellMann}), $j_2$ times $E_2$ and so on. 
With this notation, the definition of the quadratic Casimir operator (up to an unimportant common factor) is (cf. (\ref{newdefC2}))
\begin{equation}\label{Cassu3}
C_2:=F_{(n-1,1,0,...,0)}^2+F_{(n-1,0,1,0,...,0)}^2+ \cdots + F_{(n-1,0,0,0,...,1)}^2. 
\end{equation} 
Our goal now is to obtain for $C_2$ in (\ref{Cassu3}) an expression analogous to the expression (\ref{stilCasimir}) for the  $su(2)$ case. To this purpose, we calculate $F_{(n-1,1,0,...,0)}^2$, denoting by $E_k^j$ the tensor products of identities with only $E_k$ (in (\ref{GellMann})) in position $j$. Using the commutativity of $E_1^j$ and $E_1^l$, we have, 
$$
F_{(n-1,1,0,...,0)}^2=\left( \sum_{j=1}^n E_1^j \right)^2=\sum_{j=1}^n (E_1^j)^2 + 
2 \sum_{j<l} E_1^j E_1^l =\sum_{j=1}^n (E_1^j)^2+2 F_{(n-2,2,0,...,0)}, 
$$
and, using (\ref{commuanti}), we get 
$$
F_{(n-1,1,0,...,0)}^2=\frac{2n}{3} {\bf 1}+ \frac{1}{2} \sum_{j=1}^n \sum_{l=1}^8 d_{1,1}^l E_l^j +2 F_{(n-2,2,0,...,0)}. 
$$
The same calculation holds for $F_{(n-1,0,1,0,...,0)}^2$, $F_{(n-1,0,0,1,...,0)}^2$, and so on,up to 
$F_{(n-1,0,,...,0,1)}^2$, giving analogous  result as above except that the $d_{1,1}^l$ coefficients are replaced by $d_{2,2}^l$, $d_{3,3}^l$,...,$d_{8,8}^l$. 
Combining all these in (\ref{Cassu3}), we obtain:
$$
C_2=\sum_{k=1}^8 \left( \frac{2n}{3}{\bf 1} + \frac{1}{2} \sum_{j=1}^n \sum_{l=1}^8 d_{k,k}^l F_l^j\right)+2\left( F_{(n-2,2,0,...,0)}+F_{(n-2,0,2,0,...,0)}+\cdots+ F_{(n-2,0,...,0,2)}\right)=
$$   
$$
\frac{16n}{3} {\bf 1} + \frac{1}{2}\sum_{l=1}^8 \left( \sum_{k=1}^8 d_{k,k}^l \right) \left( \sum_{j=1}^n F_l^j\right) +2\left( F_{(n-2,2,0,...,0)}+F_{(n-2,0,2,0,...,0)}+\cdots+ 
F_{(n-2,0,...,0,2)}\right)=
$$
$$
\frac{16n}{3} {\bf 1} + 2\left( F_{(n-2,2,0,...,0)}+F_{(n-2,0,2,0,...,0)}+\cdots+ F_{(n-2,0,...,0,2)}\right), 
$$
where in the last equality we used (\ref{zorro}) of  Lemma \ref{help} in the Appendix. Therefore, we obtain an expression similar to the expression (\ref{stilCasimir}), that is,  
$$
C_2= \frac{16n}{3} {\bf 1} +2A, \, \texttt{ with } A:=F_{(n-2,2,0,...,0)}+F_{(n-2,0,2,0,...,0)}+\cdots+ F_{(n-2,0,...,0,2)}, 
$$
and, as in (\ref{bino1}), we obtain 
$$
C_2^k:=\left( \frac{16 n}{3} {\bf 1} + 2A \right)^k=\sum_{j=0}^k 
\begin{pmatrix} k \cr j \end{pmatrix} \left( \frac{16n}{3}\right)^{k-j} 2^j A^j. 
$$
In general, the element $C_2^k$ is a linear combination of ${\bf 1},A,A^2,...,A^{k}$and the new `direction' introduced by $C_2^k$, with respect to $\texttt{span} \{ C_2^0, C_2,...,C_2^{k-1}\}$ is in $A^k$. 

\subsection{The special case of  $n=3$ qutrits}

 According to Proposition \ref{clarif} there are three non-isomorphic irreducible representations 
which appear in the CG decomposition of $(\mathbb{C}^3)^{\otimes 3}$.   They have i-weights 
$$
(3,0,0), \quad (2,1,0), \quad (1,1,1)\equiv (0,0,0). 
$$
The module for the representation $(3,0,0)$ is the so-called {\it symmetric sector} or {\it symmetric subspace}, that is,  the subspace of $(\mathbb{C}^3)^{\otimes 3}$ spanned by vectors which are invariant under the permutation group. This 
representation has dimension $10$ and it is spanned by the (normalized) {\it Dicke states} defined as  sums of all tensor products of three among $\{|0\rangle, |1\rangle, |2\rangle\}$ which 
have equal numbers of $|0\rangle$'s, $|1\rangle$'s and $|2\rangle$'s (see subsection \ref{LAotss} below). The representation $(2,1,0)$ is 
also interesting because it is the 8-dimensional {\it adjoint representation}  (see, e.g., \cite{Knapp}). The representation $(1,1,1)$ is the {\it trivial 
(one dimensional) representation}  of $su(3)$, which is zero on every vector of a one dimensional vector space. The full Lie algebra $u^{S_3}(3^3)$ is therefore   $u^{S_3}(3^3)=u(10)\oplus u(8) \oplus u(1)$  
which has dimension $165$, so that $su^{S_3}(3^3)$ has dimension $164$. 

We are interested in calculating the dynamical Lie algebra generated by local qutrit symmetric transformations corresponding to the generators $F_{(2,1,0,...,0)}$, 
$F_{(2,0,1,0,...,0)}$,...,$F_{(2,0,0,0,...,1)}$ where the positions are occupied by the Gell-Mann matrices in (\ref{GellMann}) or the identity,\footnote{For instance $F_{(2,1,0,0,0,0,0,0)}=E_1 \otimes {\bf 1} \otimes {\bf 1}+ {\bf 1} \otimes E_1 \otimes {\bf 1}+{\bf 1} \otimes {\bf 1} \otimes E_1$}, and a symmetric two-body Hamiltonian which we take with the matrix $E_3$ in (\ref{GellMann}), that is, the two body Hamiltonian is  $F_{(1,0,0,2,0,0,0,0,0)}$ where the $2$ appears in the fourth position corresponding to $E_3$, that is, 
\begin{equation}\label{twbHam}
H:=F_{(1,0,0,2,0,0,0,0,0)}=E_3 \otimes E_3 \otimes {\bf 1}+E_3 \otimes {\bf 1} \otimes E_3 + {\bf 1} \otimes E_3 \otimes E_3. 
\end{equation} 
The generators $iF_{(2,1,0,...,0)}$, 
$iF_{(2,0,1,0,...,0)}$,...,$iF_{(2,0,0,0,...,1)}$ span a Lie algebra of {\it local} simultaneous unitary  transformations which is isomorphic to $su(3)$. We denote this Lie algebras by ${\cal L}oc$.  
Define for brevity $\hat F_j:=F_{(2,0,..,0,1,0,...,0)}$ where the $1$ appears in the $j-$th position counting the first as the zero-th position and notice that 
$$
\hat F_3^2=E_3^2\otimes {\bf 1} \otimes {\bf 1}+{\bf 1} \otimes E_3^2 \otimes {\bf 1} + {\bf 1} \otimes {\bf 1} \otimes E_3^2+ 
2 \left( E_3 \otimes E_3 \otimes {\bf 1} + E_3 \otimes {\bf 1} \otimes E_3 + {\bf 1} \otimes E_3 \otimes E_3 \right)=
$$
$$
2{\bf 1} + \frac{1}{2} \sum_{l=1}^8 d_{3,3}^l \hat F_3 
+ 2\left( E_3 \otimes E_3 \otimes {\bf 1}+E_3 \otimes {\bf 1}  \otimes E_3+ {\bf 1} \otimes E_3 \otimes E_3 \right)=
$$
$$
2 {\bf 1} +\frac{1}{\sqrt{3}} \hat F_8+ 2 F_{(1,0,0,2,0,0,0,0,0)}, 
$$
where we used the tables for the $d_{j,k}^l$ coefficients  in Appendix \ref{GM}. 
Since $ \hat F_8$ is local, the (dynamical) Lie algebra generated by  
$\{ i \hat F_3^2, {\cal L}{oc}\}$ coincides with the one generated 
by $\{ i{\bf 1}+ i  F_{(1,0,0,2,0,0,0,0,0)}, {\cal L}{oc}\}$. According 
to Theorem \ref{another}, the Lie algebra $\{i \hat F_3^2, {\cal L}oc \}_{Lie}$ contains $su^{S_3}_{cless}(27)$ (cf., (\ref{orthodec})) if and only if ${\cal L}:=\{ i H, {\cal L}oc \}_{Lie}$ does. In this case,  the two Lie algebras coincide except for their 
(one dimensional) component onto the center. They are equivalent  in terms of subspace controllability.  Therefore, instead 
of studying the Lie algebra $\{ i H, {\cal L}oc \}_{Lie}$ we shall 
study  equivalently $\{ i {\hat F}_3^2, {\cal L}oc \}_{Lie} $. In the spirit 
of Lemma \ref{lemmaneeded}, we now study the Lie algebra generated on each of the sub-representations.

\subsubsection{Lie algebra on the symmetric sector (representation 
$(3,0,0)$)}\label{LAotss}  

The most convenient basis to be chosen in the symmetric sector, which is the module of the representation $(3,0,0)$, is made of the normalized Dicke's states\footnote{We give definitions in the case $n=3$, $d=3$ that can be however naturally extended to general values of $n$ and $d$.}
$$
|\tilde \phi_{[w_1,w_2,w_3]}\rangle=\sqrt{\frac{w_1! w_2! w_3!}{3!}} |\phi_{[w_1,w_2,w_3]} \rangle, 
$$
where $|\phi_{[w_1,w_2,w_3]} \rangle$ is the sum of all states which are tensor products of $w_1$, $|0\rangle$'s,  $w_2$, $|1\rangle$'s, and $w_3$, $|2\rangle$'s, and $\sqrt{\frac{w_1! w_2! w_3!}{3!}}$ is a normalization factor. For example, we have,  
$$
|\tilde \phi_{[2,0,1]} \rangle =\frac{1}{\sqrt{3}} \left( |0 0 2\rangle+ |2 0 0 \rangle+ |0 2 0\rangle \right). 
$$
Using the definitions (\ref{Szpmstandard}) (\ref{totalOperators}), for $l=1,2$, we get, by direct calculation, 
$$
\hat S_+^l |\tilde \phi_{[w_1,w_2,w_3]}\rangle =\sqrt{(w_l+1)w_{l+1}} |\tilde \phi_{[... w_{l}+1, w_{l+1}-1,...]}, 
$$
$$
\hat S_-^l |\tilde \phi_{[w_1,w_2,w_3]}\rangle =\sqrt{(w_{l+1}+1)w_{l}} |\tilde \phi_{[... w_{l}-1, w_{l+1}+1,...]}\rangle,  
$$
which also gives the corresponding relations for $\hat S_{x,y}^l$, using (\ref{Paulimatplus}). 

Set now $l=1$, and notice that $\hat S_{\pm}^1$, $\hat S_{x,y,z}^1$ do 
not modify the number of $|2\rangle$'s appearing in the Dicke states, i.e., 
$w_3$ in $|\tilde \phi_{[w_1,w_2,w_3]} \rangle $. This fact is in line with {\bf Fact 1} of subsection \ref{duetre} if we   identify the Dicke state $|\tilde \phi_{[w_1,w_2,w_3]}$ with the Gelfand-Tsetlin state with weight vector $w=(w_1,w_2,w_3)$.  Therefore we divide the symmetric sector of three qutrits into invariant subspaces ${\cal S}_j$, $j=0,1,2,3$ where each ${\cal S}_j$  is the span of $|\tilde \phi_{w_1,w_2,j]}\rangle$. We have $\dim\left( {\cal S}_j\right)=4-j$ and $i\hat S_{x,y,z}^1$ give on  ${\cal S}_j$ an irreducible representation of $su(2)$,\footnote{Notice we have $su(2)$ here and not $su(3)$.}  a fact that is a special case of {\bf Fact 2} in subsection \ref{duetre}.

Using Theorem  \ref{theordiffformat}, since we have $i(\hat S_z^1)^2$ and $i\hat S_{x,y,z}^1$ and since ${\cal S}_0$ and ${\cal S}_2$ are CG modules, for the CG decomposition of $(\mathbb{C}^2)^{\otimes 3}$, the dynamical Lie algebra contains matrices that (in appropriate coordinates) have arbitrary blocks of dimension $4$ and $2$ with values in $su(4)$ and $su(2)$ respectively. 
Analogously,  since ${\cal S}_1$ and ${\cal S}_3$ are CG modules, for the CG decomposition of $(\mathbb{C}^2)^{\otimes 2}$,  the dynamical Lie algebra contains matrices that (in appropriate coordinates)  have arbitrary blocks of dimension $3$ and $1$ with values in $su(3)$ and $su(1)\equiv \{ 0\}$,  respectively. 
Since all these blocks have different dimensions, we can apply Lemma \ref{lemmaneeded} to conclude that the dynamical Lie algebra ${\cal L}$ contains matrices of the form (on the symmetric sector)
\begin{equation}\label{formonthe}
\begin{pmatrix} 
B_4 & 0 & 0 & 0 \cr 
0 & B_3 & 0 & 0 \cr 
0 & 0 & B_2 & 0 \cr 
0 & 0 & 0 & B_1
\end{pmatrix}
\end{equation} 
with $B_j$ arbitrary in $su(j)$. 

Let us now consider $i \hat S_{x,y,z}^2$ on the symmetric sector for which 
we have identified the above subspaces.   Direct calculation of $\hat S_+^2$ in the given basis gives that all the entries are zero except $(\hat S_+^2)_{2,5}=1$, $(\hat S_+^2)_{3,6}=\sqrt{2}$, $(\hat S_+^2)_{4,7}=\sqrt{3}$, $(\hat S_+^2)_{6,8}=\sqrt{2}$, $(\hat S_+^2)_{7,9}=2$, $(\hat S_+^2)_{9,10}=\sqrt{3}$. The matrix $\hat S_-^2$ is $(\hat S_+^2)^T$ and $\hat S_{x,y}^2$ are obtained from the `hatted' version of (\ref{rrt}). Consider for example $\hat S_x^2$, which, if we consider the block partition in (\ref{formonthe}),  has the form 
\begin{equation}\label{formaSx}
\hat S_x^2:= \begin{pmatrix}
0 & F & 0 & 0 \cr 
F^\dagger & 0 & G & 0 \cr 0 & G^\dagger & 0 & H \cr 0 & 0 & H^\dagger & 0 
\end{pmatrix}, \qquad F=\begin{pmatrix} 0 & 0 & 0 \cr 1 & 0 & 0 \cr 0 & \sqrt{2} & 0 \cr 0 & 0 & \sqrt{3} \end{pmatrix}, \quad G=\begin{pmatrix} 0 & 0 \cr \sqrt{2} & 0 \cr 0 & 2 \end{pmatrix}, 
\quad H:=\begin{pmatrix} 0 \cr  \sqrt{3} \end{pmatrix}.
\end{equation}
Doing the Lie bracket of $i\hat S_x^2$ with the matrix in (\ref{formonthe}) with 
$B_4=\texttt{diag}(i,-i,0,0)$, we obtain, up to an unimportant (nonzero) 
proportionality factor,  a matrix which has all zero entries except an  $i$ in positions $(2,5)$ and $(5,2)$. Then we can use the following fact. 
\begin{lemma}\label{generators}
Assume $n_1\geq 1$ and $n_2 \geq 1$ with $n_1+n_2 \geq 3$. 
Consider the Lie algebra ${\cal L}$ generated  by block diagonal matrices 
$\begin{pmatrix} A & 0 \cr 0 & B \end{pmatrix}$ with $A\in su(n_1)$ and $B\in su(n_2)$ arbitrary,  and by a single off diagonal matrix with $i$ in position $(j,m)$ and $(m,j)$, with 
$1\leq j \leq n_1$ and $ n_1+1 \leq m \leq n_1+n_2$.  Then ${\cal L}=su(n_1+n_2)$. 
\end{lemma}
\begin{proof}
(cf. Appendix \ref{Maxsubalg}) 
\end{proof}
\noindent This implies that all the block diagonal matrices in $su(7) \oplus su(2) \oplus su(1)$, belong to ${\cal L}$.  

We now proceed analogously by choosing $B_3=\texttt{diag}(i, -i,0)$ in (\ref{formonthe}) and taking the Lie bracket with $i\hat S_x^2$. Then we subtract a matrix which we already have proven to be in the Lie algebra above,  and we obtain a matrix with all zeros except in position $(6,7)$ and $(7,6)$, which is proportional to $i$. An application of Lemma \ref{generators} gives that ${\cal L}$ contains all block diagonal matrices in $su(9)\oplus su(1)$. Finally,  by choosing $B_4=0$, $B_3=0$, $B_2=\texttt{diag}(i,-i)$ and subtracting a matrix which was already obtained in the previous steps (that is a block diagonal matrix in $su(9) \oplus su(1)$), we have a matrix  with  zero everywhere except (up to  a proportionality factor) for the entry $(9,10)$ and $(10, 9)$ which are proportional to $i$. Applying for the last time Lemma \ref{generators}  we conclude with the following Proposition. 
\begin{proposition}\label{SSprop} (Subspace controllability on the symmetric sector)
The dynamical Lie algebra ${\cal L}$  contains matrices that are arbitrary matrices in $su(10)$ on the symmetric sector $(3,0,0)$. 
\end{proposition}

\subsubsection{Lie algebra on the adjoint representation (representation  $(2, 1, 0)$)}\label{LAar}

A basis of the adjoint representation $(2,1,0)$ can be expressed in terms of the Gelfand-Tsetlin pattern or equivalently in terms of the semi-standard Young tableaux and it was given in (\ref{A1S}).   As in the case of the symmetric sector treated above we need to put such a basis in a  one to one correspondence with a basis of the corresponding subspace in $(\mathbb{C}^3)^{\otimes 3}$ written in the computational basis of $\{ |0 \rangle, |1 \rangle, |2\rangle \}$. We do that as follows: As in the case of the symmetric sector, we associate to a state with weight vector $w=(w_1,w_2,w_3)$ a linear combination of states that,  written in the computational basis, have $w_1$ $|0\rangle $'s, $w_2$ $|1\rangle$'s and $w_3$, $|2\rangle$'s.  Furthermore, the state associate with $(w_1,w_2,w_3)$ is chosen as an eigenstate of the Young symmetrizer associated with the the Young diagram for the representation $(2,1,0)$,\footnote{Notice the similarity with the treatment for the symmetric sector} which is in this case\footnote{We are using here the cycle notation for permutation. Since the representation $(2,1,0)$ has multiplicity $2$ there are {\it two} isomorphic representations and we are considering only one. In particular isomorphic irreducible representations are labeled by {\it standard} Young tableau (SYT)  which are SSYT where each integer appears only once. In the case of the representation $(2,1,0/$ we have the SYT's  $Y_1=\young(12,3)$ and  $Y_2=\young(13,2)$. The Young symmetrizer $\Pi$ in (\ref{Youngsymmetrizer}) refers to $Y_1$.} 
\begin{equation}\label{Youngsymmetrizer}
\Pi={\bf 1 }+\left( 1\, 2 \right)-\left (1 \, 3 \right )-\left(2 \,1 \, 3 \right). 
\end{equation}
We set (cf.  (\ref{A1S})) 
\begin{equation}\label{A1SCB}
A_1 \leftrightarrow \frac{1}{\sqrt{6}} \left( 2 |0 0 1 \rangle - |1 0 0 \rangle -|0 1 0 \rangle \right);
\end{equation}
$$
A_2 \leftrightarrow \frac{1}{\sqrt{6}} \left( -2 |11 0 \rangle + |1 0 1 \rangle +|0 1 1 \rangle \right);
$$
$$
A_3 \leftrightarrow \frac{1}{\sqrt{6}} \left( 2 | 002 \rangle - |2 00 \rangle -|020 \rangle               \right)
$$
$$
A_4 \leftrightarrow \frac{1}{\sqrt{12}} \left( 2 |102 \rangle+ 2|012 \rangle -|120 \rangle - |210 \rangle - |201 \rangle - |021 \rangle  \right);
$$
$$
A_5 \leftrightarrow  \frac{1}{\sqrt{6}} \left( 2 |112 \rangle -|121\rangle - |2 11 \rangle \right); 
$$
$$
A_6\leftrightarrow \frac{1}{2} \left( |021 \rangle - |120 \rangle + |201 \rangle - |210 \rangle \right);
$$
$$
A_7\leftrightarrow \frac{1}{\sqrt{6}} \left( -2 |220 \rangle + |022\rangle + |202 \rangle  
\right)
$$
$$
A_8\leftrightarrow \frac{1}{\sqrt{6}} \left( |122 \rangle + |212 \rangle -2 |221 \rangle \right).
$$
The above rule determines the vectors in the computational basis associated with a 
given  weight vector without  ambiguity (except for a phase/normalization factor) in any case except for the weight vector $w=(1,1,1)$ for which the associate subspace (that is the space spanned by tensor products with one $|0\rangle$ one $|1\rangle$ and one $|2\rangle$) contains two linearly independent eigenvectors of the Young symmetrizer (\ref{Youngsymmetrizer}). In this case the choice is suggested by the action of the raising and lowering operators we shall consider next.\footnote{A different method to establish this correspondence is based on the use of Clebsch-Gordan coefficients \cite{Vilekin}.}

Similarly to the case of  the symmetric sector,  we shall  now 
consider the subspaces ${\cal S}_j$, $j=0,1,2$, where ${\cal S}_j$ is the span of the above vectors (\ref{A1SCB}) where there are $j$, $2$'s. Therefore 
$$
{\cal S}_{0}:=\texttt{span} \{A_1, A_2 \}, \quad {\cal S}_1:=\texttt{span}\{A_3, A_4, A_5, A_6 \},  \quad {\cal S}_2:=\texttt{span} \{ A_7, A_8 \}, 
$$
and we notice (in line with the {\bf Fact 1} of subsection \ref{duetre}) that each of these  subspaces is invariant under the action of 
$
\hat S_{\pm}^1:=S_{\pm}^1 \otimes {\bf 1} \otimes {\bf 1} + {\bf 1} \otimes S_{\pm}^1 \otimes {\bf 1} + {\bf 1} \otimes {\bf 1} \otimes S_{\pm}^1, 
$
and 
$
\hat S_{z}^1:=S_{z}^1 \otimes {\bf 1} \otimes {\bf 1} + {\bf 1} \otimes S_{z}^1 \otimes {\bf 1} + {\bf 1} \otimes {\bf 1} \otimes S_{z}^1, 
$ defined in  (\ref{Szpmstandard}), (\ref{totalOperators}). 
In particular, we have, for $l=1,...,8$, (cf. (\ref{eigeig}))  
$$
\hat S_z^1 A_l=\frac{w_1(A_l)-w_2(A_l)}{2} A_l. 
$$
On ${\cal S}_0$, 
$$
\hat S_-^1 A_1= A_2, \qquad \hat S_-^1 A_2= 0. 
$$
On ${\cal S}_1$
$$
\hat S_-^1 A_3=\sqrt{2} A_4, \qquad \hat S_-^1 A_4=\sqrt{2} A_5, \quad \hat S_-^1 A_5=0, \qquad \hat S_-^1 A_6=0.  
$$
On ${\cal S}_2$, 
$$
\hat S_-^1 A_7=A_8, \quad \hat S_-^1 A_8=0, 
$$
while $\hat S_+^1$ is the transpose. With these relations, $\texttt{span}\{A_1, A_2\}$ is the standard representation of $su(2)$, $\texttt{span}\{A_3, A_4,A_5\}$ is the $(2,0)$ (three dimensional) representation of $su(2)$, $\texttt{span}\{ A_6\}$ is the trivial representation of $su(2)$ and   $\texttt{span}\{ A_7, A_8\}$ is again the standard representation of $su(2)$, and $\hat S_z^1$ acts as $iS_z$ (for $su(2)$) on each of these representations. Applying Theorem \ref{theordiffformat} and Lemma \ref{lemmaneeded},  we find that any block matrix with the first block in $su(2)$, the second block in  $su(3)$ and the third block in $su(1)=0$ belongs to the dynamical Lie algebra. Furthermore the actions of $\hat S^1_{x,y,z}$ on $\{A_1, A_2\}$ coincide exactly with the ones on $\{A_7,A_8\}$, that is, the two standard representations of $su(2)$ coincide once we make the isomorphism $A_1 \leftrightarrow A_7$, $A_2 \leftrightarrow A_8$. Therefore, the dynamical Lie algebra contains  all  matrices of the form 
\begin{equation}
\begin{pmatrix}
B_1 & 0 & 0 & 0 \cr 
0 & B_2& 0 & 0 \cr 
0 & 0 & 0_{1 \times 1} & 0 \cr
0 & 0 & 0 & B_1
\end{pmatrix}, 
\end{equation}
with $B_1$ arbitrary in $su(2)$ and $B_2$ arbitrary in $su(3)$. 
  To these, we have to add $i\hat S_{x,y,z}^2$ and (possibly repeated) commutators.  
In the given basis, the matrix corresponding to $\hat S_x^2$ is (similarly to (\ref{formaSx})) 
\begin{equation}\label{formaSx2}
\hat S_x^2:=\begin{pmatrix}
0 & F & L & 0 \cr 
F^\dagger & 0 & 0 & G \cr L^\dagger  & 0 & 0 & H \cr 0 & G^\dagger & H^\dagger & 0 
\end{pmatrix}, \, F=\begin{pmatrix} 1& 0 & 0 \cr 0 & \frac{1}{\sqrt{2}}& 0 \end{pmatrix}, \, L:=\begin{pmatrix} 0 \cr \sqrt{\frac{3}{2}} \end{pmatrix}, \,  
G:=\begin{pmatrix} 0 & 0 \cr \frac{1}{\sqrt{2}} & 0 \cr 0 & 1 \end{pmatrix}, \, H:=\begin{pmatrix} \sqrt{\frac{3}{2}} & 0 \end{pmatrix}. 
\end{equation}
Now take the Lie bracket of $i \hat S_x^2$ with $\texttt{diag}(0,B,0,0)$ 
with $B=\begin{pmatrix} i  & 0 & 0 \cr 0 & 0 & 0 \cr 0 & 0 & -i \end{pmatrix}$ and then with 
$\texttt{diag}(0,C,0,0)$ with $C=\begin{pmatrix} i & 0 & 0 \cr 0 & -i  & 0 \cr 0 & 0 & 0 \end{pmatrix}$. We obtain,\footnote{Once again, we omit unimportant nonzero scalar factors.} a matrix which is zero everywhere except in position $(1,3)$ (and $(3,1)$)  which is proportional to  $i$. Using Lemma \ref{generators},  it follows that the dynamical Lie algebra contains all matrices of the form 
\begin{equation}\label{format56}
\begin{pmatrix} C & 0 & 0 \cr 0 & 0_{1 \times 1} & 0 \cr 0 & 0 & B \end{pmatrix}, 
\end{equation} 
with arbitrary $C \in su(5)$ and $B \in su(2)$ (cf. Lemma \ref{lemmaneeded}). For this reason,  matrices of the form  (\ref{formaSx2}) where $F$ is replaced by $0$, are also  belonging to the dynamical Lie algebra. Now consider the Lie bracket of such a matrix with a matrix in (\ref{format56}) with $C=\texttt{diag}(i,-i,0,0,0)$ and $B=0$. This gives a matrix which is zero everywhere except in position $(2,6)$ (and $(6,2)$) which is proportional 
to $i$. Again using Lemma \ref{generators} we have that the dynamical Lie algebra ${\cal{L}}$ contains all matrices of the form 
\begin{equation}\label{format57}
\begin{pmatrix}\
C & 0 \cr 0 & B
\end{pmatrix},  
\end{equation}
with $C \in su(6)$ and arbitrary $B \in su(2)$ arbitrary. Therefore, it also contains 
the matrix  of the form $i \hat S_x^2$, with  $ \hat S_x^2$ in (\ref{formaSx2}) where we replace $F$ and $L$ with the 
zero of the corresponding dimensions. Then the resulting matrix together with the matrices in (\ref{format57}) generate all of $su(8)$ from Lemma \ref{generators}. In conclusion, we have the  following proposition which is the corresponding of Proposition \ref{SSprop}, this time  for the representation $(2,1,0)$.  

\begin{proposition}\label{210prop} (Subspace controllability on the adjoint representation $(2,1,0)$)
The Lie algebra ${\cal L}$  contains matrices that are arbitrary matrices in $su(8)$ on the representation $(2,1,0)$. 
\end{proposition}

\subsubsection{Conclusion}

Combining Propositions \ref{SSprop} and \ref{210prop},  we observe that both 
$\{ (S_z^1)^2, {\cal L}oc \}$ and $\{ i H, {\cal L}oc \}$ (cf. (\ref{twbHam}))  satisfy 
the conditions of Theorems \ref{genfacLem} and \ref{another}. Using again Lemma \ref{lemmaneeded}, 
we can conclude. 

\begin{theorem}\label{teoremafinale}
The system of three qutrits with two body interaction $H$ in (\ref{twbHam})  and all local operations available has the subspace controllability property. The dynamical Lie algebra generated by all local operation is the direct sum of $su(10) \oplus su(8)$ and the span of the component of the interaction onto the Casimir Lie algebra.  

\end{theorem}

\section{Discussion and conclusions}\label{CandD}
There are several ways to decompose the state space of a network of $n$ quantum systems whose dynamics admits a group of symmetries. The method based 
on the Schur-Weyl duality advocated in \cite{conJonas} is completely general but it has the problem of finding ways to construct the so called Young symmetrizers, which are the projections onto the various invariant subspaces. These symmetrizers are  known in the case of the symmetric group (the technique based on the use of Young tableau). The problem of the computation of Young symmetrizers can be overcome in general using techniques of matrix Lie  algebras (cf. Chapter 4 in \cite{mikobook}) at the  price of manipulations  with large matrices. In our framework, for the case of the symmetric group $S_n$, we have used the fact that the decomposition into invariant subspaces for $u^{S_n}(d^n)$ coincides with the Clebsch-Gordan decomposition of $(\mathbb{C}^d)^{\otimes n}$ into irreducible representations of $su(d)$ (Theorem \ref{alwaysassumed}). This opens up the possibility of using the rich  machinery of the representation theory of $su(d)$. In particular, i-weights for irreducible representations label the various invariant subspaces,  The algebra of Casimir operators span the center of $u^{S_n}(d^n)$ (Proposition \ref{center}) and a canonical basis can be identified.  In this basis, the action of the various operators (raising and lowering operators $S_{\pm}$ and $S_z$ operator) can be calculated and take a prescribed matrix form. In these  coordinates, we can calculate the dynamical Lie algebra for a given set of quantum Hamiltonians which determines the subspace controllability property (Theorem \ref{genfacLem} and \ref{another}). In this framework, we have recast the results of \cite{confraJMP}  \cite{Marco} and proved subspace controllability (and characterized the dynamical Lie algebra) for a new case, the case of three qutrits (Theorem \ref{teoremafinale}). 

It is expected  that the techniques and facts used in the proof for the case of three qutrits can be extended and used in other cases as well as generalized up to the full case of a symmetric network of $n$ qudits for general, $n$ and $d$. In particular, in all the invariant representations,   we identified subspaces on which we could  prove subspace controllability  using the result for qubits and then extend to the full invariant subspace using different generators. This naturally suggests that the procedure might recursively  extend to higher dimensions ($d > 3$) and that a similar decomposition will be found in the more complicate  case of $n>3$ where a richer set of irreducible representations appear in the decomposition according to (\ref{clarif}). Furthermore as a consequence of Lemma \ref{lemmaneeded}, we can always  analyze one invariant subspace at a  time and prove subspace controllability on each subspace to conclude that the full dynamical Lie algebra is $u^{S_n}(d^n)$ up to its center. 

The above discussion suggests  that within the framework of this paper, one can analyze and prove subspace controllability for symmetric networks of arbitrary dimensions. A more general research plan would be to analyze dynamics and controllability for networks of quantum systems  displaying symmetries different from the one of the full symmetric group $S_n$ for which however the tools of representation theory appear to be the key to a full and satisfactory treatment. This is motivated by the current flourishing interest in the analysis of {\it structured } quantum data and their processing. This is both at the theoretical and experimental level and it is the subject of the emerging discipline of Geometric Quantum Machine Learning \cite{Marco1}, \cite{Marco2}.

\section*{Acknowledgements} The author would like to thank Marco Cerezo and Jonas Hartwig for useful and stimulating discussions on the topic of this paper.  This research was supported by the ARO MURI grant W911NF-22-S-0007.

\begin{appendix}

\section{Structure constants for $su(3)$ and some of their properties}\label{GM} 

In the following tables, the relevant values for $f_{j,k}^l$ and $d_{j,k}^l$ in (\ref{commuanti}) are reported. All the other nonzero  values can be obtained using antisymmetry (for $f_{j,k}^l$) and symmetry (for $d_{j,k}^l$).  In particular, for the non vanishing $f_{j,k}^l$, we have 
\vspace{0.25cm}
\begin{center}
\begin{tabular}{ c | c||  c | c  }
            $jkl$ & $f^{l}_{jk}$  & $jkl$ & $f^{l}_{jk}$ \\
       \hline  
            123 & 2 & 257  & 1\\
            147 & 1& 345 & 1\\ 
                156 & -1 & 367 & -1\\       
            246 & 1 & 458 & $\sqrt{3}$\\
            678 & $\sqrt{3}$ &  &  \\
\end{tabular}
\end{center}
\vspace{0.25cm}
and  the nonvanishing  $d_{j,k}^l$ are given by 
   \vspace{0.25cm}
   \begin{center}
        
\begin{tabular}{ c | c||  c | c  }
            $jkl$ & $d^{l}_{jk}$  & $jkl$ & $d^{l}_{jk}$ \\
    \hline  
            118 & $\frac{2}{\sqrt{3}}$ & 247  & -1\\
            146 & 1& 256 & 1\\ 
                157 & 1 & 338 & $\frac{2}{\sqrt{3}}$\\       
            228 & $\frac{2}{\sqrt{3}}$ & 344 & 1\\
            355 & 1 & 558 & $-\frac{1}{\sqrt{3}}$ \\
            366 & -1 & 668 & $-\frac{1}{\sqrt{3}}$ \\
            377 & -1 & 778 & $-\frac{1}{\sqrt{3}}$ \\
            448 & $-\frac{1}{\sqrt{3}}$  & 888 & $-\frac{2}{\sqrt{3}}$ \\
        \end{tabular}        
\end{center}
 \vspace{0.25cm}
 From these values one can obtain the following Lemma. 
 \begin{lemma}\label{help}
 For every $l$
 \begin{equation}\label{zorro}
 \sum_{k=1}^d d_{k,k}^l=0
 \end{equation}
 \end{lemma}
\begin{proof}
According to the above table, the only coefficients $d_{k,k}^l$ possibly different from zero are when $l=3$ or $l=8.$  Furthermore we have 
 $$
 d_{4,4}^3=d_{5,5}^3=-d_{6,6}^3=-d_{7,7}^3, 
 $$
 with all the other $d_{j,j}^3$'s equal to zero, and 
 $$
 d_{1,1}^8=d_{2,2}^8=d_{3,3}^8=-2d_{4,4}^8= -2 d_{5,5}^8= - 2 d_{6,6}^8=-2d_{7,7}^8 =-d_{8,8}^8. 
 $$
 From these, we directly verify 
$$
 \sum_{j=1}^8 d_{j,j}^3=\sum_{j=4}^7 d_{j,j}^3=0, \qquad  \sum_{j=1}^8 d_{j,j}^8=0. 
$$
 \end{proof}
 This Lemma can be generalized by considering sums 
$$
SUM:=\sum_{j=1}^8 d_{j,j}^{l_1} d_{j,j}^{l_2} \cdots d_{j,j}^{l_m}
$$
 for  $m \geq 2$  a positive integer and with  $(l_1,...,l_m)$  an $m$-ple of indexes.  The $SUM$ is always zero unless  $(l_1,...,l_m)$ contains only $3$'s and-or $8$'s. In these cases, it is still zero if it contains an odd number of $3$'s.

\section{Analysis of the quadratic Casimir eigenvalue for $su(3)$}\label{quadraticCasi}

We identify an irreducible representation of $su(3)$ by the two quantum numbers $(p,q)$ and study the quadratic Casimir eigenvalue $c_2:=c_2(p,q)$ as a function of $p$ and $q$. This is given in terms of quantum numbers $(p,q)$ (cf, e.g., formula (3.65) in \cite{Pais}\footnote{The paper \cite{Pais} uses for a representation $(m_1,m_2,m_3)$ the quantum numbers $p:=m_1-m_3$, $q:=m_2-m_3$, while we use the quantum numbers $p:=m_1-m_2$, $q:=m_2-m_3$, which explains the discrepancy between our formula (\ref{c2su3}) and formula (3.65) in \cite{Pais}.})
\begin{equation}\label{c2su3}
c_2(p,q)=p^2 +q^2 +3(p+q)+pq. 
\end{equation}
The integers $(p,q)$ are both nonnegative. We first observe from (\ref{c2su3}) that $c_2=c_2(p,q)$ is symmetric with respect to the  swapping  of $p$ and $q$. This implies that the quadratic Casimir eigenvalue is not sufficient, in general,  to distinguish between different representations, even if we restrict  ourselves to the representations appearing in the Clebsch-Gordan decomposition of $(\mathbb{C}^d)^{\otimes n}$ as described in Proposition \ref{clarif}. The following example shows this. 

\begin{example}\label{controesempio1}
Consider $n=12$. In the case  $m=9$ and $j=2$ (which is in the specified range)  in (\ref{admissrep}), we have, 
$$
(p,q)=(m-2j,j)=(5,2). 
$$
In the case $m=n=12$, $j=5$ (which is again in the specified range) we have 
$$
(p,q)=(m-2j,j)=(2,5). 
$$
Since the two pairs are connected through the exchange of $p$ and $q$, they give the same value of the quadratic Casimir eigenvalue although they correspond to different, nonisomorphic, irreducible representations.   
\end{example}

We  now outline an algorithm that, for general irreps (that is, not necessarily the ones described in Proposition \ref{clarif}),  allows one to find, for a given representation described by the quantum numbers $(p_0,q_0)$,  all the pairs $(p,q)$ which have the same value for $c_2$. 
 Because of the symmetry of the function,  it is enough to consider the area of the positive quadrant below (and including) the line $p=q$ since every point we find giving the same value for $c_2$ has a mirror image point giving the same value with respect to this line. Therefore, we shall consider only the area of the plane described by 
$$
A:=\{(p,q) \in \mathbb{N} \times \mathbb{N} \, | \, 0 \leq p, \, \, 0 \leq q \leq p \}. 
$$
Given a pair $(p_0, q_0)$,  because of the form of the function $c_2$ in (\ref{c2su3}), there is a finite number of  points $(p,q)$ such that $c_2(p,q)=c_2(p_0,q_0)$.\footnote{This will follow from the discussion below. However, a quick way to see this is to  consider the disc with radius $c_2(p_0,q_0)+\frac{9}{2}$ and centered at $\left( \frac{3}{2}, \frac{3}{2}\right)$. If $(p,q)$ is outside this disc we have 
$$c_2(p,q)=p^2+q^2+pq+3(p+q) \geq p^2 +q^2 +3\left( p+q \right)= 
\left( p+\frac{3}{2} \right)^2+ \left( q+\frac{3}{2} \right)^2-\frac{9}{2}> c_2(p_0,q_0). 
$$
Therefore only points $(p,q)$ with integer values of $p$ and $q$ inside this disc need to be considered.
} 
Calculating the partial derivative $\frac{\partial c_2}{\partial p}=2p+q+3 > 0$ it follows that $c_2$ increases along horizontal lines from left to right. Analogously, from the partial derivative with respect to $q$, it follows that $c_2$ increases along vertical lines upward. More generally, along a vector $\vec v=(1,-\alpha)$, $\alpha>0$, the directional derivative is 
$$
D_{(1,-\alpha)} c_2=(2p+3+q)-\alpha(2q+p+3). 
$$
In particular along the line with $\alpha=1$ $p>q$ implies $D_{(1,-\alpha)} c_2>0$ and the function is increasing (it is actually increasing for any $\alpha \leq 1$). Along the line with $\alpha=2$ we have $D_{(1,-\alpha)} c_2<0$.\footnote{Since $3q+3>0$, we have 
$4q+6+2p > q+3+2p$, that is $D_{(1,-2)} c_2= (q+3+2p)-2(2q+p+3)<0$.} Therefore the situation is the one  depicted in Figure \ref{figureF1}. Only points inside the two triangles have to be considered. Points below the line $q\equiv q_0$ and to the left  of the triangle $T_1$ have  to be excluded because from these points it is possible to reach the line corresponding to $\alpha=2$ and then move towards $(p_0,q_0)$ and the function will be always increasing. Therefore the starting point has a value of the function {\it strictly less} than $c_2(p_0,q_0)$. Analogously, points below the line $q \equiv q_0$ and to the right of the line corresponding to $\alpha=1$ will have a value of the function strictly higher than the value in $(p_0,q_0)$, since from $(p_0,q_0)$ it is possible to reach these points by following the line and then moving rightwards horizontally. A similar reasoning for points {\it above} the line $q \equiv q_0$ excludes all the points except the ones in the triangle $T_2$ in Figure \ref{figureF1}.

\begin{figure}[htb]
\centering
\includegraphics[width=1.3 \textwidth, height=0.55\textheight]{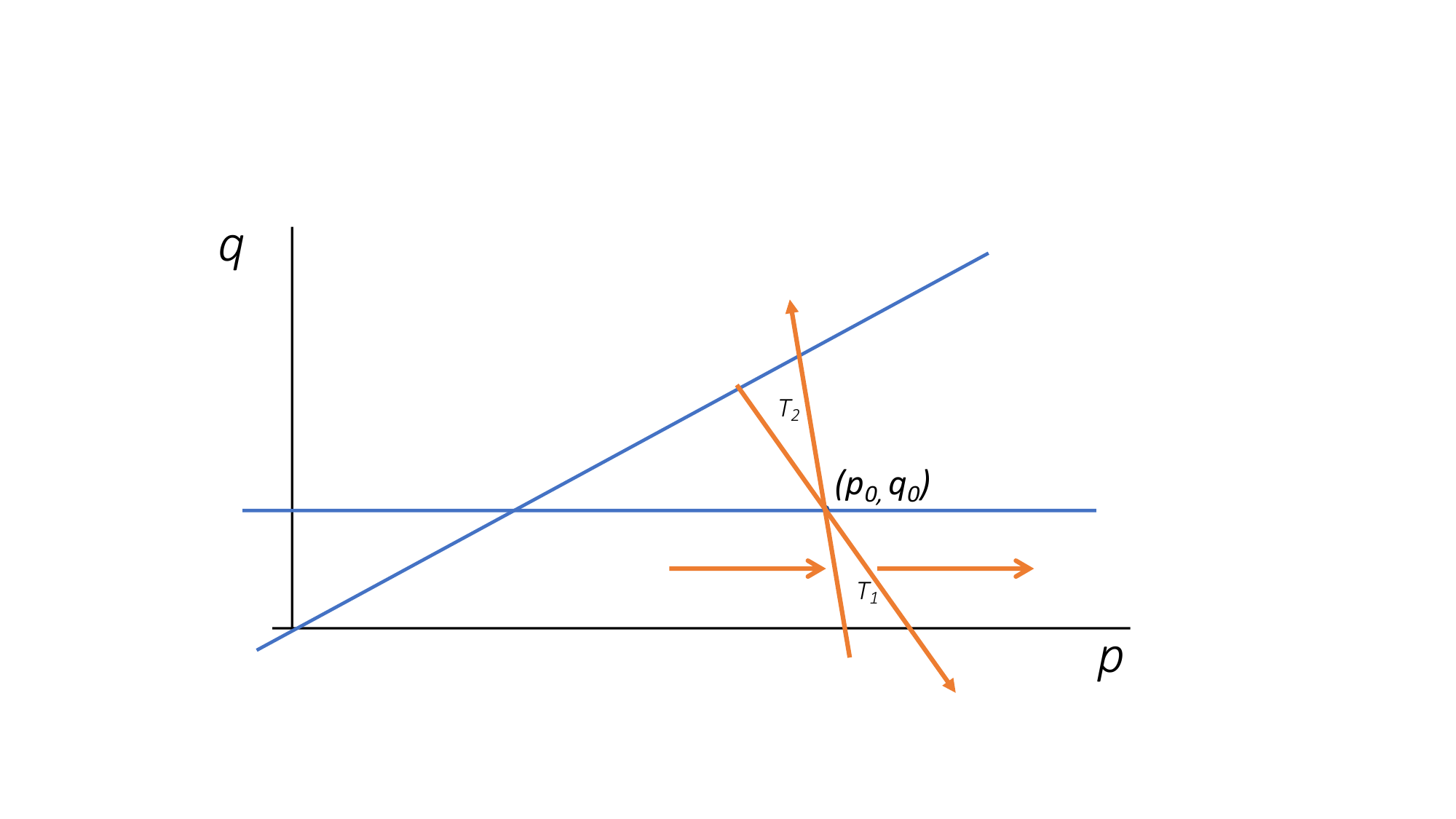}
\caption{Reducing the number of candidate 
points $(p,q)$ such that $c_2(p,q)=c_2(p_0,q_0)$.}
\label{figureF1}
\end{figure}

Having restricted to a (small) finite set the number of candidate points,   an algorithm to find all points with the same value of $c_2$ could be as follows. One checks the points in these triangles and if one  finds one point $(p_1,q_1)$ with $c_2(p_1,q_1)=c_2(p_0,q_0)$, the point $(p_1,q_1)$ also identifies two triangles. Any other other 
point has to belong to the {\it intersection}  of the two double triangles, the one for $(p_0,q_0)$ and the one for $(p_1, q_1)$. This gives two (smaller) disjoint triangles.We continue this way to eventually  reach  triangle pairs that do not have any point with the same value of $c_2$,  and the algorithm stops.

\section{Proof of Lemma \ref{generators}}\label{Maxsubalg}

\begin{proof}

Let us denote by $X_{a,b}$ ($Y_{a,b}$) a matrix which is $0$ everywhere and occupied by the matrix $i\sigma_x$ ($i\sigma_y$) at the intersection of the rows and 
columns $a$ and $b$. Then  one readily verifies\footnote{In the whole discussion 
that follows we neglect unimportant proportionality factors.} that $[X_{a,b}, Y_{a,c}]=X_{b,c}$. Applying this under the assumptions of the Lemma,  we have,  for any $k$,  $1 \leq k \leq n_1$, and $l$, $n_1+1 \leq l \leq n_1+n_2$,\footnote{Notice that  $Y_{j,k}$ and $Y_{m,l}$ are available because of the assumption that $su(n_1) \oplus su(n_2) \subseteq {\cal L}$.}
$$
[X_{j,m}, Y_{j,k}]=X_{k,m}, \qquad [X_{k,m}, Y_{m,l}]=X_{k,l}. 
$$
Since $(k,l)$ is arbitrary, we can `move' the $i$ 
everywhere in the off-diagonal blocks. Furthermore, for $1 \leq j < k \leq n_1$, denote by $Z_{j,k}$ the 
matrix which is zero everywhere except for the $j$-th element on the diagonal and the $k$-th element 
on the diagonal which are occupied by $i$ and $-i$ respectively. For, $l$,  $n_1+1 \leq l \leq n_1+n_2$, it is easily verified that $[Z_{j,k}, X_{j,l}]=Y_{j,l}$. Therefore we can change every $X_{j,l}$ 
into a $Y_{j,l}$. Thus,  every off diagonal matrix in $su(n_1+n_2)$ belongs to the Lie algebra ${\cal L}$, which also contains the block diagonal matrices with blocks in $su(n_1)$ and $su(n_2)$. This accounts for a Lie subalgebra of $su(n_1+n_2)$ of dimension 
$(n_1^2-1)+(n_2^2-1)+2n_1n_2$. 
However, for $ 1 \leq k \leq n_1$ and $n_1+1 \leq l \leq n_1+ n_2$,  we also have the linearly independent 
matrix $[X_{k,l}, Y_{k,l}]=Z_{k,l}$ which gives an extra dimension. So ${\cal L}$ is an $(n_1+n_2)^2-1$ dimensional subalgebra of $su(n_1+n_2)$ and therefore it coincides with $su(n_1+n_2)$. 
 
\end{proof}
 
 \end{appendix}

\end{document}